\documentclass[aps,prl,reprint,amssymb,amsmath,superscriptaddress,nofootinbib,longbibliography]{revtex4-2}

\usepackage{amsfonts}
\usepackage{mathtools}
\usepackage{mathrsfs}
\usepackage{amsthm}
\usepackage{comment}
\usepackage{physics}
\usepackage{graphicx}
\usepackage{xcolor}
\usepackage{enumitem}
\usepackage{hyperref}
\hypersetup{hidelinks}

\makeatletter
\let\frontmatter@footnote@produce\frontmatter@footnote@produce@footnote
\makeatother

\newtheorem{theorem}{Theorem}
\newtheorem{lemma}[theorem]{Lemma}
\newtheorem{corollary}[theorem]{Corollary}
\newtheorem{proposition}[theorem]{Proposition}
\newtheorem{definition}[theorem]{Definition}

\newcommand{\LOCC}{\mathrm{LOCC}}
\newcommand{\LSCC}{\mathrm{LSCC}}
\newcommand{\STAB}{\mathrm{STAB}}
\newcommand{\SEP}{\mathrm{SEP}}
\newcommand{\ALL}{\mathrm{ALL}}

\newcommand{\Id}{\mathbb{I}}
\newcommand{\revp}{p}
\newcommand{\revK}{K}

\providecommand{\Tr}{\operatorname{Tr}}
\providecommand{\rank}{\operatorname{rank}}

\providecommand{\PWF}{\mathrm{PWF}}

\providecommand{\Prb}{\mathbb P}
\providecommand{\dd}{\mathrm d}
\begin{document}

\title{Asymptotic Entanglement Hiding under Stabilizer Restrictions}
\author{Jicun Li}
\email{lijicun@mail.ustc.edu.cn}
\affiliation{School of Computer Science and Technology, University of Science and Technology of China, Hefei 230027, China}
\author{Wei Xie}
\thanks{\href{mailto:xxieww@ustc.edu.cn}{xxieww@ustc.edu.cn}; Corresponding author}
\affiliation{School of Computer Science and Technology, University of Science and Technology of China, Hefei 230027, China}
\affiliation{Hefei National Laboratory, University of Science and Technology of China, Hefei 230088, China}
\author{Jun Wu}
\affiliation{School of Software Engineering, University of Science and Technology of China, Hefei 230051, China}
\affiliation{Suzhou Institute for Advanced Research, University of Science and Technology of China, Suzhou 215123, China}
\author{Honglin Chen}
\affiliation{School of Computer Science and Technology, University of Science and Technology of China, Hefei 230027, China}
\author{Xiang-Yang Li}
\thanks{\href{mailto:xiangyangli@ustc.edu.cn}{xiangyangli@ustc.edu.cn}; Corresponding author}
\affiliation{School of Computer Science and Technology, University of Science and Technology of China, Hefei 230027, China}
\affiliation{Hefei National Laboratory,
	University of Science and Technology of China,
	Hefei 230088, China}

\begin{abstract}

	Entanglement is central to quantum information processing, while stabilizer operations underpin fault-tolerant quantum computation. We ask how much entanglement remains visible or distillable under stabilizer restrictions. We quantify stabilizer-visible entanglement by restricting the measured
	relative entropy of entanglement to stabilizer measurements, thereby
	obtaining converse bounds on entanglement distillation under local
	stabilizer operations and classical communication (LSCC). We demonstrate
	{\emph{magic-free asymptotic entanglement hiding}}: we construct explicit { convex mixtures of pure stabilizer states} on \(N\) qutrits per party whose unrestricted visible entanglement and \(\LOCC\)-distillable entanglement both grow as \(\Omega(N/\log N)\), while their stabilizer-visible and stabilizer-distillable entanglement vanish as \(N\to\infty\). Thus, an unbounded amount of \(\LOCC\)-distillable entanglement carried by stabilizer states can become asymptotically invisible and undistillable under stabilizer restrictions. These states also admit efficient weak simulation under LSCC, showing that
	\emph{entanglement can be asymptotically hidden in efficiently classically simulable systems}. We further prove that stabilizer-visible entanglement is \(O(1)\) with high probability for Haar-random pure states despite extensive unrestricted visibility, and vanishes uniformly over entangled Werner states as the local dimension grows through odd primes. These results reveal a fundamental separation between entanglement and magic as resources, exposing intrinsic limits on entanglement extraction using stabilizer operations.
\end{abstract}

\maketitle

\emph{Introduction.---}
Quantum entanglement is a central resource behind quantum communication, computation, and metrology \cite{bennett1993teleporting,bennett1996mixed,ekert1991quantum,shor1994algorithms,giovannetti2004quantum,horodecki2009quantum}. In the standard resource theory of entanglement, separable states are free and { local operations and classical communication} (\(\LOCC\)) define the canonical manipulation class \cite{nielsen1999conditions,vidal2000entanglement,chitambar2019quantum}. A fundamental way to quantify entanglement is through distinguishability. The relative entropy of entanglement is a standard entanglement measure and quantifies how well an entangled state can be statistically separated from the separable set \cite{vedral1997quantifying,vedral1998entanglement,plenio2001bounds,henderson2000information}.

Entanglement distillation is a central operational task in this resource theory: given many copies of a noisy bipartite state, Alice and Bob attempt to extract { nearly perfect maximally entangled states} using the allowed operations \cite{bennett1996purification,deutsch1996quantum,bennett1996concentrating,bennett1996mixed,briegel1998quantum,devetak2005distillation,horodecki2009quantum}. Relative-entropy quantities are useful not only as distinguishability measures but also as converse tools, upper-bounding the rate at which entanglement can be distilled under a specified class of operations \cite{rains2001semidefinite,fang2019non,regula2019one,lami2023no}.

In many quantum information processing settings, however, restrictions apply to the operations or measurements accessible to an observer \cite{leone2025entanglement,yanguez2025efficient}. Restricted-measurement entanglement measures formalize such restrictions on measurement access by retaining the same distinguishability task---separating an entangled state from the separable set---while optimizing only over measurements available to the observer \cite{piani2009relative}. Restrictions on local operations, including those imposed by
superselection rules and Gaussianity, are also known to limit
entanglement manipulation
\cite{bartlett2003entanglement,schuch2004nonlocal,eisert2002distilling}.

 Measurements implementable by stabilizer operations provide
	a particularly relevant restriction.  These operations---Clifford unitaries, Pauli measurements,
stabilizer ancillas, Pauli corrections, and classical
feed-forward---are central to fault-tolerant quantum computation
\cite{campbell2017roads}, admit efficient classical simulation
\cite{gottesman1997stabilizer,aaronson2004improved}, and, together with
stabilizer states, define the free sector of the resource theory of
magic
\cite{bravyi2005universal,gross2006hudson,veitch2012negative,
	mari2012positive,veitch2014resource,howard2014contextuality,
	howard2017application,wang2020efficiently}.
Recent work has uncovered limitations of stabilizer and related
classically simulable measurements in state discrimination, including
the failure to perfectly distinguish certain mutually orthogonal
stabilizer states---termed ``nonstabilizerness without magic''---and
has explored the roles of magic assistance and adaptivity
\cite{zhu2024limitations,kwon2025nonstabilizerness,
	wang2026improve,stratton2026discrimination}.
Related limitations of classically simulable operations have also been
established for universal quantum-state purification
\cite{he2026no}.
Stabilizer operations also underlie many entanglement-distillation
protocols based on stabilizer codes or syndrome measurements
\cite{matsumoto2003conversion,glancy2006entanglement,
	dur2007entanglement,shi2025stabilizer,popp2025novel,
	shi2025measurement}, and have recently been formulated as quantum
channels for general resource-state distillation
\cite{popp2026resource}.
It was recently shown that no efficient, state-agnostic LOCC protocol
can, in general, distill a constant fraction of the entanglement from
magic-dominated states \cite{gu2025magic}.
Taken together, these developments highlight both the utility and the
limitations of stabilizer operations. This motivates our central question:
\emph{Under stabilizer restrictions, how much entanglement remains visible
	or distillable?}

Our main objects are {\emph{stabilizer-visible entanglement}, defined via the stabilizer-measured relative entropy}, and distillable
entanglement under local stabilizer operations and classical
communication (LSCC).

\emph{Main result.---}
	We demonstrate \emph{magic-free asymptotic entanglement hiding}:
	we construct explicit convex mixtures of pure stabilizer states
	on \(N\) qutrits per party whose unrestricted visible entanglement
	and \(\LOCC\)-distillable entanglement both grow as
	\(\Omega(N/\log N)\), while their stabilizer-visible and
	stabilizer-distillable entanglement vanish as \(N\to\infty\).
	LSCC protocols acting on these states admit efficient weak
	simulation on a classical computer.
This shows that a classically simulable system can remain almost blind to the large amount of entanglement it contains. {Figure~\ref{fig:distillation-separation} illustrates this separation.}

\begin{figure}[t]
\centering
\includegraphics[width=\columnwidth]{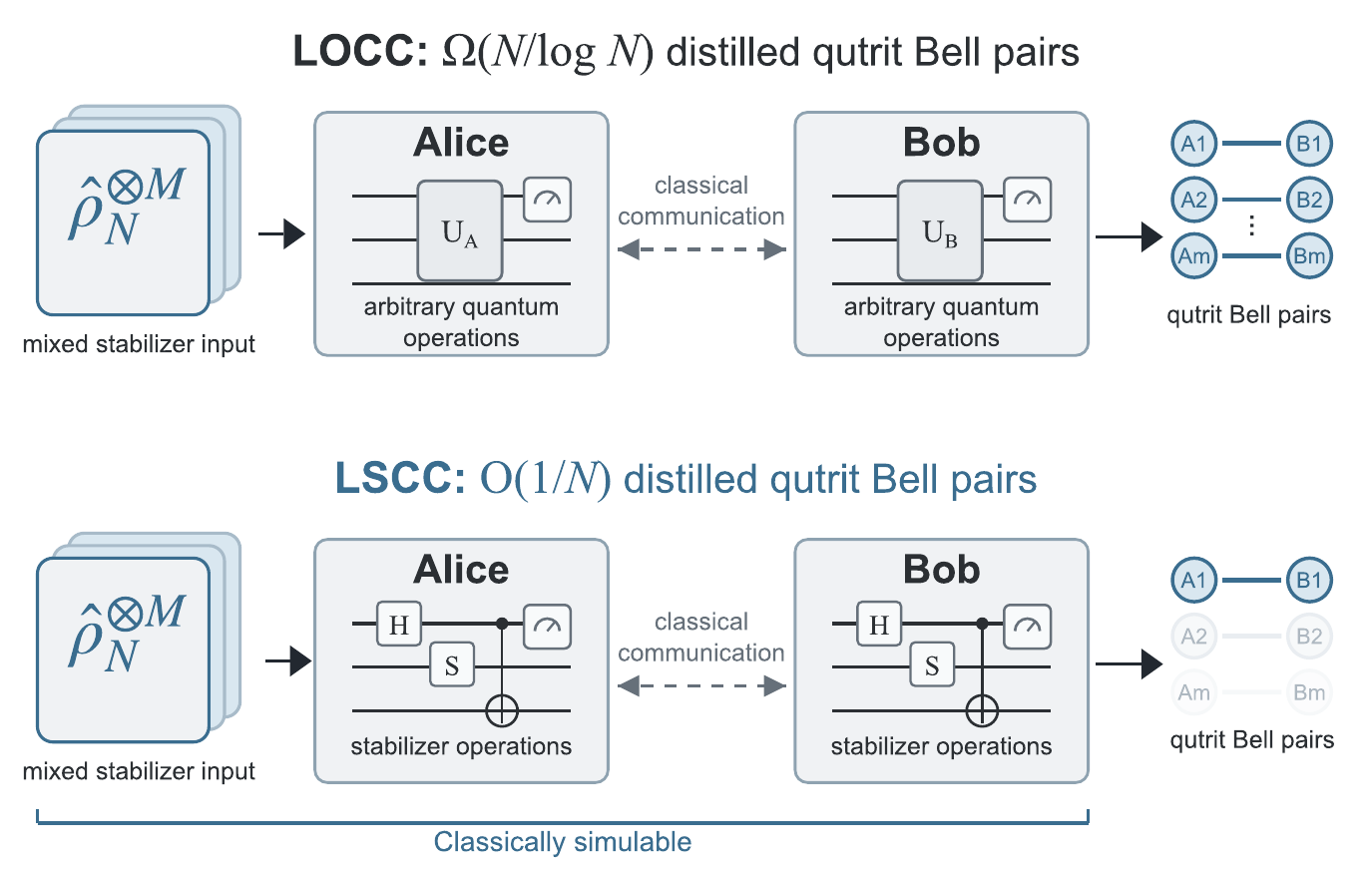}
\caption{{Alice and Bob share \(M\) copies of the mixed stabilizer state
		\(\widehat{\rho}_N\), with \(N\) qutrits per party in each copy.
		LOCC achieves an asymptotic distillation yield of
		\(\Omega(N/\log N)\) per input copy, whereas under LSCC
		the yield is bounded above by \(O(1/N)\).}}
\label{fig:distillation-separation}
\end{figure}

\emph{Stabilizer-visible entanglement.---}
We call a POVM $\mathcal{M}=\{E_i\}_{i\in I}$ a stabilizer measurement if it can be implemented using stabilizer-ancilla preparation, Clifford unitaries, adaptive Pauli measurements, classical randomness, and classical post-processing, all of which are free operations in the resource theory of magic \cite{gross2006hudson,mari2012positive,veitch2014resource}. We denote this measurement class by $\STAB$. Operationally, $\STAB$ captures precisely the statistical information available to an observer equipped only with stabilizer operations.

{ Given any state \(\rho\),} let $\mu_{\mathcal M}(\rho)=\{\Tr(E_i\rho)\}_i$ denote the outcome distribution induced by a measurement $\mathcal M$. { We write \(D_{\rm KL}\) for the classical KL divergence and \(\SEP\) for the set of bipartite separable states.} { We define the stabilizer-measured relative entropy and the corresponding stabilizer-visible entanglement by}
\begin{equation}
\label{eq:stab-visible-defs}
\begin{aligned}
D_{\STAB}(\rho\Vert\sigma)
&:=\sup_{\mathcal M\in\STAB}
D_{\rm KL}\!\left(\mu_{\mathcal M}(\rho)\Vert\mu_{\mathcal M}(\sigma)\right),\\
E_{\STAB}(\rho)&:=\inf_{\sigma\in\SEP}D_{\STAB}(\rho\Vert\sigma).
\end{aligned}
\end{equation}
This quantity provides a stabilizer-restricted analogue of the measured relative entropy of entanglement, in the spirit of restricted-measurement entanglement measures \cite{piani2009relative,vedral1997quantifying}. 

To quantify what is lost by restricting the measurement class to { stabilizer measurements}, we compare \(D_{\STAB}\) with the standard measured relative entropy, namely the same expression optimized over { all POVMs, whose class we denote by \(\ALL\)} \cite{piani2009relative,matthews2009distinguishability}:
\begin{equation}
\label{eq:all-visible-defs}
\begin{aligned}
D_{\ALL}(\rho\Vert\sigma)
&:=\sup_{\mathcal M\in{\ALL}}
D_{\rm KL}\!\left(\mu_{\mathcal M}(\rho)\Vert\mu_{\mathcal M}(\sigma)\right),\\
E_{\ALL}(\rho)&:=\inf_{\sigma\in\SEP}D_{\ALL}(\rho\Vert\sigma).
\end{aligned}
\end{equation}
Thus \(E_{\STAB}\) and \(E_{\ALL}\) quantify the same task---distinguishing a state from the separable set---under stabilizer-restricted and unrestricted measurements, respectively. Throughout this paper, logarithms are base two; all relative entropies and rates are measured in bits.

{ Our first result shows that stabilizer measurements change the scaling of visible entanglement for typical high-dimensional states. Let \(p\) be any fixed prime, let \(\ket{\psi_n}\) be Haar-random on
\(\mathbb C^{p^n}\otimes\mathbb C^{p^n}\), and write \(\psi_n:=\ket{\psi_n}\!\bra{\psi_n}\). Standard results on the typicality of Haar-random bipartite pure states~\cite{page1993average,hayden2006aspects,nechita2007asymptotics} imply}
\begin{equation*}
E_{\ALL}(\psi_n)=n\log p+O(1)
\end{equation*}
{ with probability \(1-o(1)\) as \(n\to\infty\). The theorem below shows that stabilizer measurements change this scaling.}

\begin{theorem}\label{thm:haar-visible}
{
For each fixed prime \(p\), there exist constants \(C_p,c_p>0\) such that, for Haar-random
\(\ket{\psi_n}\) on \(\mathbb C^{p^n}\otimes\mathbb C^{p^n}\),
\begin{equation}
\label{eq:haar-visible-bound}
E_{\STAB}(\psi_n)\le C_p
\end{equation}
with probability at least \(1-e^{-c_p n^2}\).
}
\end{theorem}

{In fact, the proof yields the stronger result
\(D_{\STAB}(\psi_n\Vert\Id/p^{2n})\le C_p\)
on the same high-probability event, showing that the
stabilizer-measured relative entropy to the maximally mixed state remains
bounded independently of \(n\).}

 Consequently, the stabilizer-visible fraction is \(O(1/n)\): typical high-dimensional states contain extensive entanglement, while stabilizer measurements access only a dimension-independent amount.

{
Werner states form a canonical one-parameter family of bipartite mixed states. 
Their \(U\otimes U\) symmetry makes many entanglement questions analytically tractable. Werner’s construction showed that entanglement does not imply Bell nonlocality \cite{werner1989quantum,vollbrecht2001entanglement,horodecki2009quantum}. 
Let \(p\) be any odd prime and let \(F_p\) denote the swap operator on \(\mathbb C^p\otimes\mathbb C^p\). Define
\begin{equation}
\label{eq:werner-state}
\begin{aligned}
\rho_{\mathrm W,p}(t)&=t\,\tau_{+,p}+(1-t)\tau_{-,p},\\
\tau_{\pm,p}&=\frac{{\Id}\pm F_p}{p(p\pm1)}.
\end{aligned}
\end{equation}
The Werner state \(\rho_{\mathrm W,p}(t)\) is entangled for \(0\le t<1/2\) and separable for \(1/2\le t\le1\) \cite{werner1989quantum,horodecki2009quantum}. 
For fixed \(t<1/2\), the unrestricted visible entanglement is dimension independent:
\begin{equation}
\label{eq:werner-all}
E_{\ALL}\!\left(\rho_{\mathrm W,p}(t)\right)
=
t\log(2t)+(1-t)\log(2(1-t)).
\end{equation}
The right-hand side is the known relative entropy of entanglement
\(E_R(\rho_{\mathrm W,p}(t))\)~\cite{vollbrecht2001entanglement};
the equality \(E_{\ALL}=E_R\) follows from a simple two-outcome
measurement argument given in the Supplemental Material~\cite{supplemental}.
Stabilizer-visible entanglement behaves differently.

\begin{proposition}\label{prop:werner}
For every odd prime \(p\) and \(0\le t<1/2\),
\begin{equation}
\label{eq:werner-stab}
\begin{aligned}
E_{\STAB}\!\left(\rho_{\mathrm W,p}(t)\right)=
&\frac{2t}{p+1}\log(2t) \\
&+\frac{p+1-2t}{p+1}
\log\!\left(\frac{p+1-2t}{p}\right).
\end{aligned}
\end{equation}
In particular,
\begin{equation*}
\sup_{0\le t\le1/2}
E_{\STAB}\!\left(\rho_{\mathrm W,p}(t)\right)
\le \log\!\left(1+\frac{1}{p}\right)
\longrightarrow0
\end{equation*}
as \(p\to\infty\) through odd primes.
\end{proposition}

Comparing Eqs.~\eqref{eq:werner-all} and \eqref{eq:werner-stab}, the unrestricted value remains finite for fixed \(t<1/2\), whereas the stabilizer-visible entanglement vanishes uniformly along odd-prime local dimensions. 
The bound in Proposition~\ref{prop:werner} is uniform in \(t\), showing that the collapse is not confined to a special point of the Werner family.
}

A more striking question is whether the hiding can occur for stabilizer states themselves. For pure bipartite stabilizer states, the answer is no: local Clifford operations transform any such state into { maximally entangled pairs} and local product states, so { local stabilizer operations fully reveal and extract its entanglement} \cite{fattal2004entanglement}. We show that the mixed-state case is different.

\begin{theorem}\label{thm:mixed-simple}
In \(3\otimes 3\), let
\(\ket{\Psi^+_{12}}=(\ket{11}+\ket{22})/\sqrt2\) and
\(\rho_* = \frac13\ket{00}\bra{00} + \frac23\ket{\Psi^+_{12}}\bra{\Psi^+_{12}}\).
Then \(\rho_*\) is a mixed stabilizer state and satisfies
\[
E_{\STAB}(\rho_*) < E_{\ALL}(\rho_*).
\]
\end{theorem}

\emph{Sketch of proof.---}
First, \(\rho_*\) is explicitly a convex mixture of stabilizer states. Let
\(\omega=e^{2\pi i/3}\). Then
\[
\rho_*=\frac13\sum_{a=0}^{2}\ket{{\xi}_a}\bra{{\xi}_a},\qquad
\ket{{\xi}_a}=\frac1{\sqrt3}\sum_{j=0}^{2}\omega^{a j^2}\ket{jj},
\]
and each \(\ket{{\xi}_a}\) is a two-qutrit stabilizer state.  Second, \(\rho_*\)
is maximally correlated.
{ Here \(\Delta\) denotes the complete dephasing channel in the local
computational bases, and \(D\) denotes the quantum relative entropy.}  With
\(\sigma_*=\Delta(\rho_*)=\frac13\sum_i\ket{ii}\bra{ii}\), { the
unrestricted value is}
\[
E_{\ALL}(\rho_*)=D(\rho_*\Vert\sigma_*)=\frac23.
\]
  Third, we use Gross's discrete Wigner representation to obtain an
  efficiently computable Wigner--KL upper bound on \(D_{\STAB}\). For odd-prime local
  dimension, stabilizer measurements admit a nonnegative representation
  in this formalism and therefore induce stochastic maps from phase space
  to measurement outcomes. Consequently, for stabilizer states
  \(\rho_*\) and \(\sigma_*\), the classical data-processing inequality gives
\[
D_{\STAB}(\rho_*\Vert\sigma_*)
\le D_{\rm KL}(W_{\rho_*}\Vert W_{\sigma_*}).
\]
Since
\(\rho_*\) and \(\sigma_*\) are mixtures of stabilizer states, this gives
\begin{align*}
E_{\STAB}(\rho_*)
&\le D_{\STAB}(\rho_*\Vert\sigma_*)\\
&\le D_{\rm KL}\!\left(W_{\rho_*}\Vert W_{\sigma_*}\right)\\
&=\frac13\log 3
<\frac23.
\end{align*}
This proves the claimed strict gap.\hfill\(\square\)

Thus stabilizer-restricted visibility is not a consequence of input magic: it already occurs for a mixed state that is free in the resource theory of magic.

The central construction below shows that this hidden entanglement can dominate even when
the total unrestricted visible entanglement diverges, and does so using mixed
stabilizer input states.

\begin{theorem}
\label{thm:mixed-visible}
There exists a sequence of mixed stabilizer states
\(\widehat\rho_N\in\mathcal D((\mathbb C^3)^{\otimes N}_A\otimes(\mathbb C^3)^{\otimes N}_B)\) such that
\begin{equation}
\label{eq:visible-separation}
\begin{aligned}
E_{\STAB}(\widehat\rho_N)
&\le \frac{\log 3}{N}\longrightarrow0,\\
E_{\ALL}(\widehat\rho_N)
&=\Omega\!\left(\frac{N}{\log N}\right)\longrightarrow\infty.
\end{aligned}
\end{equation}
\end{theorem}

\emph{Sketch of proof.---}
For each \(b\), we construct a mixed-stabilizer
maximally correlated block \(\rho_b\) on \(b\) qutrits per party.  The block \(\rho_b\) is
an equal mixture of \(3^b\) maximally correlated stabilizer states obtained from
one vector in each of \(3^b\) mutually unbiased stabilizer bases
\cite{wootters1989optimal,gibbons2004discrete,gross2006hudson}.  The only properties
used here are that these stabilizer states have flat computational-basis
amplitudes and Wigner supports with the prescribed spread intersection pattern.
The dephased state \(\Delta(\rho_b)\) is separable. For this block,
\begin{equation*}
E_{\ALL}(\rho_b)=D(\rho_b\Vert\Delta(\rho_b))=1-3^{-b},
\end{equation*}
because the maximally correlated variational witness is feasible against all
separable states.  On the stabilizer side, the Wigner--KL upper bound gives
\begin{equation*}
E_{\STAB}(\rho_b)
\le D_{\rm KL}\!\left(W_{\rho_b}\Vert W_{\Delta(\rho_b)}\right)
= b\log 3/3^b .
\end{equation*}
Set \(b_N=\lceil2\log_3 N\rceil\), \(r_N=\lfloor N/b_N\rfloor\), and
\begin{equation}
\label{eq:rhohat-construction}
\widehat\rho_N=
\rho_{b_N}^{\otimes r_N}
\otimes(\ket{00}\bra{00})^{\otimes(N-b_Nr_N)}.
\end{equation}
The unrestricted lower bound is obtained by measuring only the nontrivial
blocks; the maximally correlated witness tensorizes for this product
construction, giving
\begin{equation*}
E_{\ALL}(\widehat\rho_N)
\ge r_N(1-3^{-b_N})=
\Omega(N/\log N).
\end{equation*}
For the stabilizer upper bound, use the product separable reference
\(\Delta(\rho_{b_N})^{\otimes r_N}\) and the product stabilizer padding:
\begin{equation*}
E_{\STAB}(\widehat\rho_N)
\le r_N b_N\log 3/3^{b_N}
\le \log 3/N .
\end{equation*}
This proves both the vanishing stabilizer-visible bound and the diverging
unrestricted bound.\hfill\(\square\)

{ Theorem~\ref{thm:mixed-visible} therefore shows that, even for states that are free in the resource theory of magic, an unbounded amount of entanglement visible to unrestricted measurements can become asymptotically invisible to stabilizer measurements.}

\emph{Stabilizer entanglement distillation.---}
We now turn from witnessing entanglement to extracting it. Many practical entanglement
distillation schemes are built from stabilizer codes and syndrome measurements
\cite{bennett1996mixed,dur2007entanglement,matsumoto2003conversion,glancy2006entanglement,shi2025stabilizer,popp2025novel,shi2025measurement}.
We analyze the information-theoretic limits of this stabilizer toolbox by
comparing distillable entanglement under unrestricted \(\LOCC\) with distillable
entanglement under local stabilizer operations and classical communication. Let
$\LSCC$ denote local stabilizer operations and classical communication: Alice
and Bob may use local Clifford circuits, stabilizer ancillas, Pauli
measurements, Pauli corrections, classical feed-forward, and classical
communication. For a state \(\rho\) on tensor powers of \(p\)-dimensional
qudits, define the one-shot \(\epsilon\)-error { \(\LSCC\)} distillable entanglement as
\begin{equation}
\begin{aligned}
E_{\LSCC}^{D,(1),\epsilon}(\rho):=
\sup\{m\log p:\ &m\in\mathbb Z_{\ge0},\ \exists\Lambda\in\LSCC,\\
&{{\frac12\left\|\Lambda(\rho)-\Phi_p^{\otimes m}\right\|_1\le\epsilon}}\}.
\end{aligned}
\end{equation}
Here
\(\Phi_p=\ket{\Phi_p}\!\bra{\Phi_p}\), with
\(\ket{\Phi_p}=p^{-1/2}\sum_{j=0}^{p-1}\ket{jj}\), is the standard
{ \(p\)-dimensional maximally entangled state}; the factor \(\log p\) records the output in
bits. Replacing \(\LSCC\) by \(\LOCC\) gives
\(E_{\LOCC}^{D,(1),\epsilon}\).
This quantity quantifies how much entanglement can be converted into a
maximally entangled state by local stabilizer protocols.

For Haar-random pure states, the unrestricted one-shot \(\LOCC\) yield
remains extensive: for every fixed \(0<\epsilon<1\),
\begin{equation*}
	E_{\LOCC}^{D,(1),\epsilon}(\psi_n)
	={{n\log p+O_{p,\epsilon}(1)}},
\end{equation*}
with probability tending to one as \(n\to\infty\)
\cite{nechita2007asymptotics,buscemi2013general}; stabilizer restrictions change this
behavior qualitatively.

\begin{corollary}\label{thm:haar-distill}
{
For every fixed prime \(p\) and every fixed \(0<\epsilon<1\), there exist constants
\(C_{p,\epsilon},c_p>0\) such that, for Haar-random \(\ket{\psi_n}\) on
\(\mathbb C^{p^n}\otimes\mathbb C^{p^n}\), with
\(\psi_n:=\ket{\psi_n}\!\bra{\psi_n}\),
\begin{equation}
\label{eq:haar-distill-bound}
E_{\LSCC}^{D,(1),\epsilon}(\psi_n)\le C_{p,\epsilon}
\end{equation}
with probability at least \(1-e^{-c_p n^2}\).
}
\end{corollary}

{ Equation~\eqref{eq:haar-distill-bound} shows that stabilizer access not only hides most of the entanglement from measurement, but also prevents its extraction into copies of \({\Phi}_p\) at an extensive one-shot rate.} The corollary follows directly from Theorem~\ref{thm:haar-visible}, using the monotonicity of \(E_{\STAB}\) under \(\LSCC\) and a Bell stabilizer test for the distilled output; see the Supplemental Material~\cite{supplemental}.

For asymptotic distillation, define
\begin{equation}
\label{eq:asymp-lscc}
E_{\LSCC}^{D}(\rho):=
\lim_{\epsilon\to0}\liminf_{M\to\infty}
\frac1M E_{\LSCC}^{D,(1),\epsilon}(\rho^{\otimes M}),
\end{equation}
with \(E_{\LOCC}^{D}\) defined analogously. The same mixed-stabilizer construction yields an asymptotic
separation for entanglement distillation.

\begin{theorem}
\label{thm:distillation-separation}
{ For the mixed stabilizer states $\widehat\rho_N$ defined in Eq.~\eqref{eq:rhohat-construction},}
\begin{equation}
\label{eq:distill-separation}
\begin{aligned}
E_{\LOCC}^{D}(\widehat\rho_N)
&=\Omega\!\left(\frac{N}{\log N}\right){\longrightarrow\infty},\\
E_{\LSCC}^{D}(\widehat\rho_N)
&\le\frac{\log 3}{N}{\longrightarrow0}.
\end{aligned}
\end{equation}
\end{theorem}

\emph{Sketch of proof.---}
The \(\LOCC\) lower bound follows from the same blocks used in
Theorem~\ref{thm:mixed-visible}. For each block, the coherent information is
\(I(A\rangle B)_{\rho_b}=1-3^{-b}\); hence the hashing inequality
\cite{devetak2005distillation} gives
\[
E_{\LOCC}^{D}(\rho_b)\ge 1-3^{-b}.
\]
Distilling the \(r_N\) blocks independently
therefore yields
\[
E_{\LOCC}^{D}(\widehat\rho_N)\ge r_N(1-3^{-b_N})=\Omega(N/\log N).
\]
For the \(\LSCC\) converse, we adapt the standard relative-entropy
method for bounding entanglement-distillation rates
\cite{vedral1998entanglement,rains2001semidefinite,lami2024distillable}
to stabilizer-restricted measurements.  Monotonicity under \(\LSCC\),
together with a robust normalization bound for approximate maximally
entangled targets, yields
\begin{equation}
	\label{eq:asymp-converse-main}
	E_{\LSCC}^{D}(\rho)
	\le
	\limsup_{M\to\infty}\frac{1}{M}
	E_{\STAB}(\rho^{\otimes M}).
\end{equation}
For \(\rho=\widehat\rho_N\), the regularized Wigner--KL upper bound proved in
the Supplemental Material~\cite{supplemental} gives
\[
\limsup_{M\to\infty}\frac{1}{M}
E_{\STAB}(\widehat\rho_N^{\otimes M})
\le
\frac{r_Nb_N\log 3}{3^{b_N}}
\le \frac{\log 3}{N},
\]
which proves the stated \(\LSCC\) upper bound.\hfill\(\square\)

This result is the distillation counterpart of the
visible-entanglement separation in
Eq.~\eqref{eq:visible-separation}: stabilizer restrictions can
limit not only how much entanglement is visible but also how much
can be distilled, even when the input states are free in the resource
theory of magic.

Since \(\widehat\rho_N\) is a uniform mixture of pure stabilizer states
with efficiently samplable labels and efficiently constructible tableaux,
LSCC protocols acting on it admit efficient weak simulation on a classical computer. Nevertheless, this classically simulable system carries unbounded entanglement, while entanglement accessible under classically simulable operations vanishes as \(N\to\infty\).
This coexistence demonstrates that
\emph{entanglement can be asymptotically hidden in efficiently classically
	simulable systems}.

\emph{Conclusion.---}
Our results demonstrate that stabilizer restrictions can severely limit operational access to entanglement, even for magic-free states. For our explicit family of mixed stabilizer states, both the entanglement visible to unrestricted measurements and the entanglement distillable by LOCC grow as $\Omega(N/\log N)$, whereas their stabilizer-restricted counterparts vanish asymptotically. LSCC protocols acting on our explicit family admit efficient weak
simulation on a classical computer. Together with the entanglement
separation, this demonstrates that
\emph{entanglement can be asymptotically hidden in efficiently
	classically simulable systems}.
We further establish analogous limitations for Haar-random pure states and Werner states, covering typical high-dimensional pure states and a highly symmetric mixed-state family.

Since many quantum codes support fault-tolerant implementations of stabilizer operations, our results reveal a fundamental information-theoretic limitation on entanglement distillation using these operations alone. More broadly, our results establish a fundamental operational separation between entanglement and magic: magic-free states can carry an unbounded amount of entanglement that becomes asymptotically inaccessible within the stabilizer subtheory. An important direction for future work is to identify the structural origin of this gap and determine the minimum nonstabilizer resource required to close it.

\emph{Acknowledgments.---} This work was partially supported by the Quantum Science and Technology--National Science and Technology Major Project (Grant No.~2021ZD0302901) and the National Natural Science Foundation of China (Grant No.~62672467 and No.~62102388). We thank OpenAI's ChatGPT (GPT-5) for assistance with language editing
and figure preparation, as well as with strengthening the rigor and
checking the details of the proofs of
Lemmas~\ref{lem:normal-form} and~\ref{lem:counting}.

\bibliography{Asymptotic_Entanglement_Hiding_under_Stabilizer_Restrictions}

\clearpage
\onecolumngrid
\setcounter{equation}{0}
\setcounter{figure}{0}
\setcounter{table}{0}
\setcounter{theorem}{0}
\renewcommand{\theequation}{S\arabic{equation}}
\renewcommand{\thefigure}{S\arabic{figure}}
\renewcommand{\thetable}{S\arabic{table}}
\renewcommand{\thetheorem}{S\arabic{theorem}}
\renewcommand{\theHequation}{S.\arabic{equation}}
\renewcommand{\theHfigure}{S.\arabic{figure}}
\renewcommand{\theHtable}{S.\arabic{table}}
\renewcommand{\theHtheorem}{S.\arabic{theorem}}

\begin{center}
{\large\bfseries Supplemental Material for ``Asymptotic Entanglement Hiding under Stabilizer Restrictions''}
\end{center}

\section{Relation to previous work}

\emph{Restricted entanglement theory.—}
Restricted-measurement relative entropies of entanglement based on
\(\LOCC\), separable (\(\SEP\)), and PPT measurements have been studied
extensively
\cite{piani2009relative,li2014relative,berta2024entanglement,
	lami2024distillable}.
Restrictions to Gaussian operations and to computationally
efficient protocols and measurements have also been studied
\cite{eisert2002distilling,leone2025entanglement,yanguez2025efficient}.
Additionally, cone-restricted information theory replaced the
positive-semidefinite cone in one-shot entropy programs by restricted
convex cones \cite{george2024cone}.
Our quantity \(E_{\STAB}\) lies within Piani's framework, with the
measurement class restricted to stabilizer measurements. Our
primary focus, however, is the resulting operational mismatch between
the resource theories of entanglement and magic.

\emph{Stabilizer state discrimination under stabilizer restrictions.—}
Kwon introduced ``nonstabilizerness without magic'' by exhibiting
mutually orthogonal stabilizer states that cannot be perfectly
distinguished by stabilizer operations \cite{kwon2025nonstabilizerness}.
Ao \emph{et al.} traced such free-state discrimination gaps to a general
convex-geometric mechanism and established their persistence in
sequential asymptotic discrimination \cite{ao2026resourcefulness}.
By contrast, our work studies how stabilizer restrictions constrain
operational access to entanglement, revealing a fundamental separation
between entanglement and magic as resources.

\emph{Interplay between entanglement and magic resource theory.—}
Gu \emph{et al.} identified a computational transition between
entanglement-dominated and magic-dominated regimes: the entanglement
tasks considered there admit efficient, state-agnostic protocols in the
former but become intractable in the latter \cite{gu2025magic}.
Our setting instead compares the operational power of unrestricted LOCC
and LSCC without imposing efficiency constraints. We demonstrate
magic-free asymptotic entanglement hiding, thereby establishing a
resource-theoretic separation between entanglement and magic, distinct
from their computational separation between state regimes.
More recently, Ganardi \emph{et al.} developed a framework for
heterogeneous local resource theories over quantum networks
\cite{ganardi2026manipulating}.
\(\LSCC\) fits this framework when the local resource is magic, but its
general bounds do not capture our Bell-pair distillation setting:
the local-resource bound is trivial because the target marginals are
free, whereas the entanglement bound ignores the stabilizer
restriction. Our converse captures both constraints jointly. Recently, a distributed resource theory of entanglement and magic
was introduced, showing that states with neither entanglement nor magic
can nevertheless be resourceful under LSCC~\cite{yuan2026distributed}.
Our work addresses a different operational task: entanglement distillation
under LSCC.

\section{Definitions and notation}

\subsection{Stabilizer formalism}
\label{sec:stab-formalism}

{ We use the standard stabilizer formalism for qudits of prime
	local dimension \(p\), following Refs.~\cite{gross2006hudson,veitch2014resource}.}
For a single qudit of prime local dimension \(\revp\), let
\(\omega_\revp=e^{2\pi i/\revp}\) and define
\[
X(a)\ket x=\ket{x+a},\qquad
Z(b)\ket x=\omega_\revp^{bx}\ket x,
\qquad a,b,x\in\mathbb F_\revp .
\]
For \(k\) qudits we use the tensor-product notation
\[
X(a)=X(a_1)\otimes\cdots\otimes X(a_k),
\qquad
Z(b)=Z(b_1)\otimes\cdots\otimes Z(b_k),
\qquad a,b\in\mathbb F_\revp^k .
\]
The \(k\)-qudit Pauli group is
\[
\mathcal P_k=\{\omega_\revp^c Z(b)X(a):a,b\in\mathbb F_\revp^k,\ c\in\mathbb F_\revp\}.
\]
{For \(\revp=2\), replace the phase factor \(\omega_\revp^c\)
by \(i^c\), with \(c=0,1,2,3\).}
A Clifford unitary is a unitary normalizing this Pauli group,
\[
U\mathcal P_kU^\dagger=\mathcal P_k .
\]
A pure stabilizer state is a state of the form
{
\[
\ket\phi=U\ket{0}^{\otimes k},
\]
}
where \(U\) is a Clifford unitary.  
A mixed stabilizer state is a convex combination of
pure stabilizer states {\cite{gross2006hudson,veitch2014resource}}. Our use of this term is not restricted to normalized projectors onto stabilizer subspaces.

\begin{definition}[Stabilizer operations]
\label{def:stab-operation}
{
A stabilizer operation is a quantum channel implementable by a finite adaptive
protocol composed of (i) adjoining stabilizer-state ancillas, (ii) Clifford
unitaries, (iii) Pauli measurements, (iv) discarding subsystems, and (v)
classical randomness and feed-forward conditioned on previous measurement
outcomes \cite{veitch2014resource}.
}
\end{definition}

{ In the context of fault-tolerant quantum computation and the resource theory of
	magic, stabilizer operations form a standard free toolbox; \(\STAB\) is the corresponding
	restricted measurement class.}
{ Retaining and classically processing the measurement record
of such a protocol gives the restricted measurement class used below.}

\begin{definition}[Stabilizer measurements]
\label{def:stab-measurement}
{
A finite-outcome POVM is called a stabilizer measurement if its outcome
distribution can be realized by a stabilizer circuit with classical control:
stabilizer-state preparation, Clifford unitaries, adaptive Pauli measurements
whose later choices may depend on earlier outcomes, discarding of systems, and
classical feed-forward and post-processing of measurement outcomes
\cite{veitch2014resource}. The class of all such POVMs is denoted by
\(\STAB\).
}
\end{definition}

\subsection{Wigner functions}
\label{sec:stab-wigner-conventions}

{ The positive-Wigner comparison uses Gross's discrete
phase-space representation \cite{gross2006hudson}.}
For \(k\) qudits of odd-prime local dimension, write the discrete phase space as
\[
V_k=(\mathbb F_\revp^2)^k=\mathbb F_\revp^k\oplus\mathbb F_\revp^k .
\]
For \(u=(a,b)\in V_k\), define the Gross displacement operators
\[
T_u=\omega_\revp^{-2^{-1}a\cdot b}Z(b)X(a),
\]
where \(2^{-1}\) is the inverse of \(2\) in \(\mathbb F_\revp\).  The Gross
phase-point operators are
\begin{equation}
\label{eq:sm-phase-point-operators}
A_0=\revp^{-k}\sum_{v\in V_k}T_v,
\qquad
A_u=T_uA_0T_u^\dagger .
\end{equation}
The Wigner function of a state \(\rho\) and the dual Wigner representation of a
POVM element, or effect, \(E\) are
\begin{equation}
\label{eq:sm-wigner-functions}
W_\rho(u)=\revp^{-k}\Tr(A_u\rho),\qquad
W_E(u)=\Tr(A_uE).
\end{equation}
We use the following standard properties of this representation:
\begin{enumerate}[label=(\roman*)]
\item The phase-point operators satisfy
\begin{equation}
\label{eq:sm-phase-point-identities}
\Tr A_u=1,\qquad
\Tr(A_uA_v)=\revp^k\delta_{u,v},\qquad
\sum_{u\in V_k}A_u=\revp^k{\Id} .
\end{equation}
\item The probability of the effect \(E\) on the state \(\rho\) is
\begin{equation}
\label{eq:sm-wigner-born-rule}
\Tr(E\rho)=\sum_{u\in V_k}W_E(u)W_\rho(u).
\end{equation}
	\item Clifford unitaries act by affine symplectic permutations of phase space:
	for each Clifford \(U\) there is a permutation \(\pi_U\) such that
	\(U^\dagger A_uU=A_{\pi_U(u)}\).  Consequently, Clifford conjugation only
	permutes Wigner distributions and preserves classical KL divergence.
	\item In odd prime dimension, stabilizer states and stabilizer effects have
	{ nonnegative} Wigner functions
	{\cite{gross2006hudson,veitch2012negative}}.
\end{enumerate}

A quantum state \(\rho\) is positive-Wigner, abbreviated PWF, if and only if
\(W_\rho(u)\ge0\) for every phase-space point \(u\).  A POVM element \(E\) is
PWF if and only if \(W_E(u)\ge0\) for every \(u\).  Let \(\PWF\) denote the
class of finite-outcome POVMs whose effects are all PWF.  In the
odd-prime-power phase-space setting,
\begin{equation}
\label{eq:sm-measurement-class-inclusion}
\STAB\subseteq\PWF\subseteq\ALL .
\end{equation}
{ The first inclusion follows from Wigner positivity of
stabilizer protocols
\cite{veitch2012negative,mari2012positive}.}
Here \(\ALL\) denotes the class of all finite-outcome POVMs.

\subsection{Restricted measured relative entropy}

{ We compare unrestricted measurements with stabilizer measurements through
measured relative entropy.}  A finite-outcome POVM
\(\mathcal M=\{E_y\}_{y\in\mathcal Y}\) maps a state \(\rho\) to the classical
outcome distribution
\[
\mu_{\mathcal M}(\rho)(y)=\Tr(E_y\rho).
\]
The distinguishability of two such outcome distributions is measured by the
classical relative entropy, or KL divergence,
\[
D_{\rm KL}(P\Vert Q)=\sum_y P(y)\log\frac{P(y)}{Q(y)},
\]
with the usual conventions \(0\log(0/q)=0\) and
\(p\log(p/0)=+\infty\) for \(p>0\).
All logarithms are base two.
Optimizing this classical distinguishability over an allowed measurement class
gives the measured relative entropy, originally studied in quantum information by
Donald and by Hiai and Petz \cite{donald1986relative,hiai1991proper}.  For
restricted measurement families, this viewpoint underlies distinguishability
under measurement restrictions \cite{matthews2009distinguishability}.  { For
stabilizer measurements, we define \(D_{\STAB}\) by restricting the same optimization
to \(\STAB\).}
\begin{definition}[Measured relative entropy under restricted measurements]
\label{def:measured-relative-entropy}
	For states \(\rho,\sigma\), define the unrestricted measured relative entropy
	\begin{equation}
	\label{eq:sm-unrestricted-measured-relative-entropy}
	D_{\ALL}(\rho\Vert\sigma)
	:=\sup_{\mathcal M\in\ALL}
	D_{\rm KL}\!\left(\mu_{\mathcal M}(\rho)\Vert
	\mu_{\mathcal M}(\sigma)\right),
	\end{equation}
	and define the stabilizer-measured relative entropy
	\begin{equation}
	\label{eq:sm-stabilizer-measured-relative-entropy}
	D_{\STAB}(\rho\Vert\sigma)
	:=\sup_{\mathcal M\in\STAB}
	D_{\rm KL}\!\left(\mu_{\mathcal M}(\rho)\Vert
	\mu_{\mathcal M}(\sigma)\right).
	\end{equation}
\end{definition}

Applying measured relative entropy to distinguish \(\rho\) from the separable
set gives a measured relative entropy of entanglement, in the
restricted-measurement sense of Piani \cite{piani2009relative}.  The following quantities quantify how well a state can be
distinguished from the separable set using the allowed
measurements on a single copy.
\begin{definition}[Visible entanglement]
\label{def:visible-entanglement}
	For a bipartite state \(\rho\), the unrestricted and stabilizer-visible
	entanglement quantities are
	\begin{equation}
	\label{eq:sm-visible-entanglement}
	\begin{aligned}
	E_{\ALL}(\rho)
	&:=D_{\ALL}(\rho\Vert\SEP)
	:=\inf_{\sigma\in\SEP}D_{\ALL}(\rho\Vert\sigma),\\
	E_{\STAB}(\rho)
	&:=D_{\STAB}(\rho\Vert\SEP)
	:=\inf_{\sigma\in\SEP}D_{\STAB}(\rho\Vert\sigma).
	\end{aligned}
	\end{equation}
\end{definition}

We also use the usual relative entropy of entanglement,
\[
E_R(\rho):=D(\rho\Vert\SEP)
:=\inf_{\sigma\in\SEP}D(\rho\Vert\sigma),
\]
where \(D\) denotes quantum relative entropy.

{ We record the unrestricted Haar result used
in the main text. It is an immediate consequence of Nechita's largest-eigenvalue asymptotics
for induced random states \cite{nechita2007asymptotics}.}
\begin{proposition}[Unrestricted Haar visible entanglement]
	\label{prop:haar-all}
	{ Let \(\psi:=\ket\psi\!\bra\psi\), where \(\ket\psi\) is Haar-random on \(\mathbb C^d\otimes\mathbb C^d\). Then}
	\begin{equation}
	\label{eq:sm-haar-unrestricted-visibility}
	E_{\ALL}(\psi)=\log d+O(1)
	\end{equation}
	{ with probability \(1-o(1)\) as \(d\to\infty\).}
\end{proposition}

\begin{proof}
The upper bound is immediate:
		\(E_{\ALL}(\psi)\le E_R(\psi)\le\log d.\)
	
	For the lower bound, { let \(\psi_A:=\Tr_B\psi\) and}
	measure \(\{\Pi_\psi,{\Id}-\Pi_\psi\}\), where
	\(\Pi_\psi=|\psi\rangle\langle\psi|\).  For any separable \(\sigma\),
	\[
	D_{\rm KL}\big((1,0)\,\|\,(\Tr\Pi_\psi\sigma,1-\Tr\Pi_\psi\sigma)\big)
	=
	-\log \Tr(\Pi_\psi\sigma).
	\]
	Moreover
	\[
	\sup_{\sigma\in\SEP}\Tr(\Pi_\psi\sigma)
	=
	\lambda_{\max}(\psi_A),
	\]
	where \(\lambda_{\max}(\psi_A)\) is the largest squared Schmidt coefficient of \(\psi\).  Therefore
	\[
	E_{\ALL}(\psi)\ge -\log \lambda_{\max}(\psi_A).
	\]
	{ By Nechita's largest-eigenvalue asymptotics
	{\cite{nechita2007asymptotics}},}
	\(d\,\lambda_{\max}(\psi_A)\to 4\) almost surely.  Thus
	\[
	E_{\ALL}(\psi)\ge \log d-O(1)
	\]
	{ with probability \(1-o(1)\) as \(d\to\infty\).}  Combining the two bounds gives the claim.
\end{proof}

\subsection{Entanglement distillation}

Entanglement distillation asks how many { nearly perfect maximally entangled states} can be
extracted from a bipartite input state using local operations and classical
communication {(\(\LOCC\))}.  We also use local stabilizer operations and classical
communication (\(\LSCC\)): the subclass of {\(\LOCC\)} protocols in which each party's
local quantum steps belong to the stabilizer toolbox above, with classical
communication.

\begin{definition}[One-shot distillable entanglement]
\label{def:distill-one-shot}
	Fix a prime local dimension \(p\), and let
	\[
	\Phi_p=\ket{\Phi_p}\!\bra{\Phi_p},
	\qquad
\ket{\Phi_p}=p^{-1/2}\sum_{j=0}^{p-1}\ket{jj}.
	\]
	For an operation class \(\mathcal C\in\{\LOCC,\LSCC\}\), the one-shot
	\(\epsilon\)-error \(\mathcal C\)-distillable entanglement of a state \(\rho\),
	measured in bits {\cite{fang2019non,regula2019one}}, is
	\begin{equation}
	\label{eq:sm-one-shot-distillable-entanglement}
	E_{\mathcal C}^{D,(1),\epsilon}(\rho)=
	\sup\left\{m\log p:m\in\mathbb Z_{\ge0},\ \exists\Lambda\in\mathcal C,
	\frac12\|\Lambda(\rho)-{\Phi_p^{\otimes m}}\|_1\le\epsilon\right\}.
	\end{equation}
	Here \(\Lambda\) ranges over channels from the input systems to \(m\) output
	qudits for Alice and \(m\) output qudits for Bob.
\end{definition}

\begin{definition}[Asymptotic distillable entanglement]
\label{def:stab-distill}
	For the same operation class \(\mathcal C\in\{\LOCC,\LSCC\}\), the asymptotic
	\(\mathcal C\)-distillable entanglement is the optimal { i.i.d.}\ rate with
	vanishing error {\cite{devetak2005distillation}}:
	\begin{equation}
	\label{eq:sm-asymptotic-distillable-entanglement}
	E_{\mathcal C}^{D}(\rho)=\lim_{\epsilon\to0}\liminf_{M\to\infty}
	\frac1M E_{\mathcal C}^{D,(1),\epsilon}(\rho^{\otimes M}).
	\end{equation}
	The cases \(\mathcal C=\LOCC\) and \(\mathcal C=\LSCC\) give the usual { \(\LOCC\)}-distillable entanglement and the stabilizer-distillable entanglement,
	respectively.
\end{definition}
Thus \(E_{\LSCC}^{D}\) is the operational analogue of stabilizer-visible
entanglement: it asks { how much entanglement remains distillable} when
each local quantum operation is restricted to the stabilizer toolbox.

{ We record the unrestricted Haar one-shot distillation
result. It follows directly from the pure-state one-shot distillation bound
of Buscemi and Datta \cite{buscemi2013general}, combined with Nechita's
largest-eigenvalue asymptotics \cite{nechita2007asymptotics}.}
\begin{proposition}[Unrestricted Haar one-shot distillation]
	\label{prop:haar-locc-distill}
	Fix a prime \(p\) and \(0<\epsilon<1\). For Haar-random
	\(\ket{\psi_n}\) on
	\(\mathbb C^{p^n}\otimes\mathbb C^{p^n}\),
	\begin{equation}
	\label{eq:sm-haar-unrestricted-distillation}
		E_{\LOCC}^{D,(1),\epsilon}(\psi_n)
		=n\log p+O_{p,\epsilon}(1)
	\end{equation}
	with probability \(1-o(1)\) as \(n\to\infty\).
\end{proposition}

\begin{proof}
	{
	Set \(d=p^n\). For any achievable \(m\), write
	\(\tau=\Lambda(\psi_n)\). Since \(\LOCC\) cannot increase Schmidt number,
	\(\tau\) has Schmidt number at most \(p^n\). Since
	\(\Phi_p^{\otimes m}\) has Schmidt rank \(p^m\),
	\[
	\Tr\!\left(\Phi_p^{\otimes m}\tau\right)\le p^{n-m}.
	\]
	The trace-distance condition gives
	\[
	\Tr\!\left(\Phi_p^{\otimes m}\tau\right)\ge1-\epsilon.
	\]
	Therefore
	\[
	m\log p\le n\log p-\log(1-\epsilon),
	\]
	and hence
	\[
	E_{\LOCC}^{D,(1),\epsilon}(\psi_n)
	\le n\log p-\log(1-\epsilon)
	=n\log p+O_\epsilon(1).
	\]

	For the lower bound, set
	\(m_n:=\lfloor-\log_p\lambda_{\max}(\psi_{n,A})\rfloor\).
	Nielsen's pure-state conversion criterion {\cite{nielsen1999conditions}}
	gives an exact \(\LOCC\) conversion to
	\(\Phi_p^{\otimes m_n}\), and hence
	\[
	E_{\LOCC}^{D,(1),\epsilon}(\psi_n)
	\ge m_n\log p
	\ge-\log\lambda_{\max}(\psi_{n,A})-\log p.
	\]
	Nechita's largest-eigenvalue asymptotics
	{\cite{nechita2007asymptotics}} give
	\(d\,\lambda_{\max}(\psi_{n,A})\to4\) almost surely, so the right-hand side
	is \(n\log p-O_p(1)\) with probability \(1-o(1)\). Combining the two bounds
	proves the claim.
	}
\end{proof}

\section{Proof of Theorem 1}

{ In this section, we prove the main-text bound in
Eq.~\eqref{eq:haar-visible-bound}. The unrestricted case was established
in Proposition~\ref{prop:haar-all}; here we prove the complementary
stabilizer-restricted bound.}

Fix a prime \(\revp\).  Let
\[
\mathcal H=(\mathbb C^\revp)^{\otimes L},
\qquad
\revK=\dim\mathcal H=\revp^L .
\]
For the bipartite application in the main text,
\[
\mathcal H=\mathbb C^{\revp^n}\otimes\mathbb C^{\revp^n},
\qquad
L=2n,
\qquad
\revK=\revp^{2n}.
\]
\begin{theorem}[Theorem~1 of the main text]
	\label{thm:haar-stab-main}
	Fix a prime \(\revp\).  There are constants \(C_\revp,c_\revp>0\), depending only on
	\(\revp\), such that { for Haar-random \(\ket\psi\) on
	\(\mathbb C^{\revp^n}\otimes\mathbb C^{\revp^n}\), with \(\psi:=\ket\psi\!\bra\psi\)},
	\begin{equation}
	\label{eq:sm-haar-stabilizer-visible}
	\Prb\left[E_{\STAB}(\psi)\le C_\revp\right]
	\ge
	1-\exp[-c_\revp\log^2\revK],
	\qquad
	\revK=\revp^{2n}.
	\end{equation}
\end{theorem}

\emph{Sketch of proof.---}
The intuition is that a Haar-random state cannot concentrate substantial
weight on any low-dimensional subspace accessible to a stabilizer
measurement.  We formalize this intuition in four steps.

\emph{(i) Separable reference.---}
The maximally mixed state
\({\Id}/\revK
=({\Id}_{\revp^n}/\revp^n)
\otimes({\Id}_{\revp^n}/\revp^n)\)
is separable.  Hence
\(E_{\STAB}(\psi)\le
D_{\STAB}(\psi\Vert{\Id}/\revK)\), so it suffices to bound the
distinguishability of \(\psi\) from the maximally mixed state under
stabilizer measurements.

\emph{(ii) Normal-form reduction.---}
By Lemma~\ref{lem:normal-form}, based on
Ref.~\cite[Theorem~4]{heimendahl2020axiomatic}, every stabilizer
measurement is a convex mixture of classical post-processings of rank-one
adaptive Pauli measurements.  Since mixing and classical post-processing
cannot increase the relevant relative entropy, it is enough to bound every
such rank-one measurement.

\emph{(iii) { Ordered probabilities and low-rank accepting projectors}.---}
Fix a rank-one adaptive stabilizer measurement
\(\mathcal B=\{\pi_i\}_{i=1}^{\revK}\) and set
\(p_i:=\Tr(\pi_i\psi)\).  Since the maximally mixed state produces the
uniform distribution,
\begin{equation}
\label{eq:sm-basis-kl-deficit}
D_{\rm KL}\!\left(
\mu_{\mathcal B}(\psi)
\middle\Vert
\mu_{\mathcal B}({\Id}/\revK)
\right)
=
\sum_{i=1}^{\revK}p_i\log(\revK p_i).
\end{equation}
Order the probabilities as
\(p_{(1)}\ge\cdots\ge p_{(\revK)}\).  For \(k>\revK/2\), normalization gives
\(p_{(k)}\le1/k\le2/\revK\), so this large-index tail contributes at most
\(\log2\).  It therefore remains to control the leading probabilities
\(p_{(k)}\) with \(k\le\revK/2\).

Let \(\pi_{(i)}\) be the outcome projector corresponding to \(p_{(i)}\).
Coarse-graining the \(k\) most likely outcomes gives the rank-\(k\)
accepting projector
\[
P_k:=\sum_{i=1}^k\pi_{(i)},
\qquad
\bra{\psi}P_k\ket{\psi}
=\sum_{i=1}^k p_{(i)}.
\]

We first count the low-rank accepting projectors that can arise from
rank-one adaptive Pauli-measurement trees.  Lemma~\ref{lem:counting} shows
that the logarithm of the number of such projectors with rank at most \(r\)
is bounded by \(C_\revp r\log^2(e\revK/r)\).
Lemma~\ref{lem:uniform-low-rank} then combines this counting estimate with
Haar concentration to show that, with high probability, no low-rank
accepting subspace captures substantially more weight from \(\psi\) than
its dimension fraction, up to a logarithmic factor.  More precisely, with
probability at least \(1-\exp[-c_\revp\log^2\revK]\),
\begin{equation}
\label{eq:sm-low-rank-envelope-overview}
\bra{\psi}P\ket{\psi}
\le
A_\revp\frac{r}{\revK}\log^2\frac{e\revK}{r}
\end{equation}
holds simultaneously for every admissible rank-\(r\) accepting projector
with \(1\le r\le\revK/2\). Thus we have
\begin{equation}
\label{eq:sm-ordered-probability-envelope}
p_{(k)}
\le
\frac1k\sum_{i=1}^k p_{(i)}
\le
\frac{A_\revp}{\revK}\log^2\frac{e\revK}{k}, \quad 1\le {k}\le\revK/2.
\end{equation}

\emph{(iv) Estimating the sum.---}
Substituting this probability envelope into the entropy deficit gives
\begin{equation}
\label{eq:sm-entropy-deficit-summation}
\sum_{k\le\revK/2}p_{(k)}\log(\revK p_{(k)})
\le
\frac{A_\revp}{\revK}
\sum_{k\le\revK/2}
L_k^2\log_+(A_\revp L_k^2)
=O_\revp(1),
\qquad
L_k:=\log\frac{e\revK}{k},
\end{equation}
where \(\log_+x:=\max\{0,\log x\}\); the last sum is bounded by the
corresponding finite integral.  Together with the \(\log2\) tail bound,
this gives a uniform \(O_\revp(1)\) bound for every rank-one adaptive
stabilizer measurement.  The normal-form reduction and classical data
processing then imply
\(E_{\STAB}(\psi)\le C_\revp\) on the same high-probability event.
\hfill\(\square\)

A rank-one adaptive stabilizer measurement on \(L\) qudits is described by an adaptive Pauli-measurement tree in which each complete branch measures a maximal independent family of pairwise commuting Pauli operators. Its leaves form an orthonormal stabilizer basis
\[
\mathcal B=\{\pi_i\}_{i=1}^\revK,
\qquad
\pi_i=\ket{\phi_i}\bra{\phi_i}.
\]
For \(A\subseteq\{1,\ldots,\revK\}\), the associated accepting projector is
\[
P_A=\sum_{i\in A}\pi_i .
\]
Let \(\mathcal P_{\rm ad}\) denote the set of all such accepting projectors,
over all rank-one adaptive stabilizer measurements and all subsets \(A\).

The normal-form reduction states that any stabilizer measurement can be reduced,
up to classical randomness and post-processing, to a rank-one adaptive
Pauli-measurement tree.
\begin{lemma}
	\label{lem:normal-form}
	Let \(\revp\) be a prime and let
	\(\mathcal H=(\mathbb C^\revp)^{\otimes L}\), \(\revK=\revp^L\).
	For every finite-outcome stabilizer measurement
	\(\mathcal M=\{E_y\}_{y\in\mathcal Y}\), there exist probabilities
	\(\lambda_\theta\), rank-one adaptive stabilizer measurements
	\[
	\mathcal B_\theta=\{\pi_{\theta,i}\}_{i=1}^\revK,
	\]
	and stochastic maps \(T_\theta(y|i)\) such that, for every state \(\rho\),
	\begin{equation}
	\label{eq:sm-normal-form-outcomes}
	\mu_{\mathcal M}(\rho)(y)
	=
	\sum_\theta \lambda_\theta
	\sum_{i=1}^\revK
	T_\theta(y|i)\Tr(\pi_{\theta,i}\rho).
	\end{equation}
	Consequently, for any states \(\rho,\tau\),
	\begin{equation}
	\label{eq:sm-normal-form-data-processing}
	D_{\rm KL}\!\left(
	\mu_{\mathcal M}(\rho)
	\middle\Vert
	\mu_{\mathcal M}(\tau)
	\right)
	\le
	\sum_\theta \lambda_\theta
	D_{\rm KL}\!\left(
	\mu_{\mathcal B_\theta}(\rho)
	\middle\Vert
	\mu_{\mathcal B_\theta}(\tau)
	\right),
	\end{equation}
	with the usual extended-value convention.
\end{lemma}
\begin{proof}
	Fix a stabilizer-protocol realization of
	\(\mathcal M=\{E_y\}_{y\in\mathcal Y}\).  Before its terminal classical
	post-processing, retain the complete fine-grained transcript
	\(x\in\mathcal X\), including the external random seed and all adaptive
	Pauli-measurement outcomes, and let
	\(\{\mathcal J_x\}_{x\in\mathcal X}\) be the corresponding stabilizer
	instrument.  The terminal post-processing is then a stochastic map
	\(q(y|x)\), and hence
	\[
	\mu_{\mathcal M}(\rho)(y)
	=
	\sum_{x\in\mathcal X}
	q(y|x)\Tr\mathcal J_x(\rho).
	\]  

	We apply the constructive normal-form reduction of
	Ref.~\cite{heimendahl2020axiomatic} while retaining the transcript
	labels.  Here and below, \(x\) labels a transcript of the original
	protocol, whereas \(a\) labels a terminal branch of the current reduced
	protocol.
	
	At every stage, write the component label as
	\(\theta=(s,\mathbf z)\), where \(s\) is the original external random seed
	and \(\mathbf z\) records the Pauli outcomes externalized by the preceding
	normal-form elimination steps.  For each component \(\theta\), let
	\(\lambda_\theta\) be its input-independent mixing weight,
	\(\mathcal A_\theta\) its set of terminal branches, and
	\(\{\mathcal J_{\theta,a}\}_{a\in\mathcal A_\theta}\) its fine-grained
	instrument.  We maintain a deterministic reconstruction map
	\(
	g_\theta:\mathcal A_\theta\longrightarrow\mathcal X
	\)
	such that
	\begin{equation}
		\label{eq:sm-normal-form-transcript}
		\mathcal J_x
		=
		\sum_\theta\lambda_\theta
		\sum_{\substack{a\in\mathcal A_\theta\\g_\theta(a)=x}}
		\mathcal J_{\theta,a},
		\qquad x\in\mathcal X.
	\end{equation}
	
	Initially, \(\mathbf z=\varnothing\),
	\(\lambda_{(s,\varnothing)}\) is the probability of the original seed
	\(s\), and
	\(\{\mathcal J_{(s,\varnothing),a}\}_a\) is the fine-grained instrument
	conditioned on that seed.  The map
	\(g_{(s,\varnothing)}\) restores the complete original transcript by
	adjoining \(s\) to the Pauli-measurement record \(a\).  Hence
	Eq.~\eqref{eq:sm-normal-form-transcript} holds initially.
	
	Consider one elimination step in a component \(\theta\) of weight
	\(\lambda_\theta\).  On a branch passing through the eliminated node, the
	outcome \(z\in\mathbb F_\revp\) is uniform and input-independent.  The
	reduction replaces \(\theta\) by \((\theta,z)\), with
	\(\lambda_{(\theta,z)}=\lambda_\theta/\revp\), and replaces the branch
	Kraus operator \(K_z\) by
	\(K'_z=\sqrt{\revp}\,K_z\).  The new reconstruction map reinserts \(z\) at
	the eliminated node of the original transcript.  Consequently, the
	weighted contribution of the affected branch is unchanged:
	\[
	\frac{\lambda_\theta}{\revp}
	K'_z\rho K_z^{\prime\dagger}
	=
	\lambda_\theta K_z\rho K_z^\dagger.
	\]
	For a branch that bypasses the eliminated node, both its Kraus operator and
	its reconstructed transcript are independent of \(z\).  Hence the
	\(\revp\) copies of weight \(\lambda_\theta/\revp\) reproduce the original
	contribution:
	\[
	\sum_{z\in\mathbb F_\revp}
	\frac{\lambda_\theta}{\revp}K\rho K^\dagger
	=
	\lambda_\theta K\rho K^\dagger.
	\]
	Applying any branchwise partial trace preserves both identities, so they
	hold at the level of the fine-grained CP maps
	\(\mathcal J_{\theta,a}\).  Thus every branchwise replacement in
	Eqs.~(28)--(32) of Ref.~\cite{heimendahl2020axiomatic} preserves
	Eq.~\eqref{eq:sm-normal-form-transcript}.  Iterating the reduction gives
	the claimed label-preserving convex decomposition.

	At the end of the reduction, for each fixed \(\theta\), the remaining branches
	\(a\) are associated with mutually orthogonal stabilizer-code projectors
	\(P_{\theta,a}\) satisfying
	\[
	P_{\theta,a}P_{\theta,b}=0\quad(a\neq b),
	\qquad
	\sum_a P_{\theta,a}={\Id}.
	\]
	The conditional Clifford dilation and discarding following branch \(a\) are
	trace preserving on the corresponding subnormalized post-measurement state,
	so
	\[
	\Tr\mathcal J_{\theta,a}(\rho)
	=
	\Tr(P_{\theta,a}\rho P_{\theta,a})
	=
	\Tr(P_{\theta,a}\rho).
	\]
	Define
	\(
	S_\theta(y|a):=q\bigl(y\,|\,g_\theta(a)\bigr).
	\)
	Taking traces in Eq.~\eqref{eq:sm-normal-form-transcript}
	and substituting into the initial expression for
	\(\mu_{\mathcal M}(\rho)(y)\) gives
	\begin{equation}
	\label{eq:sm-normal-form-code-projectors}
	\mu_{\mathcal M}(\rho)(y)
	=
	\sum_\theta\lambda_\theta
	\sum_a S_\theta(y|a)\Tr(P_{\theta,a}\rho).
	\end{equation}
	
	For each \((\theta,a)\), complete the commuting Pauli constraints defining \(P_{\theta,a}\) to a maximal commuting Pauli family.  Measuring the added Paulis below the leaf \(a\) gives a rank-one refinement
	\[
	P_{\theta,a}
	=
	\sum_{i\in I_{\theta,a}}\pi_{\theta,i}.
	\]
	Since the completion may depend on \(a\), the collection
	\(
	\mathcal B_\theta
	=
	\{\pi_{\theta,i}\}_{i=1}^{\revK}
	\)
	is a rank-one adaptive stabilizer measurement.  For \(i\in I_{\theta,a}\), define \(T_\theta(y|i):=S_\theta(y|a)\).  Substituting the refinement into the preceding expression gives
	\[
	\mu_{\mathcal M}(\rho)(y)
	=
	\sum_\theta\lambda_\theta
	\sum_{i=1}^{\revK}
	T_\theta(y|i)\Tr(\pi_{\theta,i}\rho).
	\]

	Finally, define
	\[
	r_\rho(\theta,i)
	:=
	\lambda_\theta\Tr(\pi_{\theta,i}\rho),
	\qquad
	r_\tau(\theta,i)
	:=
	\lambda_\theta\Tr(\pi_{\theta,i}\tau).
	\]
	The classical channel \((\theta,i)\mapsto y\), with transition probabilities \(T_\theta(y|i)\), sends \(r_\rho\) and \(r_\tau\) to \(\mu_{\mathcal M}(\rho)\) and \(\mu_{\mathcal M}(\tau)\), respectively.  Hence, by classical data processing and the fact that both joint distributions have the same marginal \(\lambda_\theta\),
	\[
	\begin{aligned}
		D_{\rm KL}\!\left(
		\mu_{\mathcal M}(\rho)
		\middle\Vert
		\mu_{\mathcal M}(\tau)
		\right)
		&\leq
		D_{\rm KL}(r_\rho\Vert r_\tau)\\
		&=
		\sum_\theta\lambda_\theta
		D_{\rm KL}\!\left(
		\mu_{\mathcal B_\theta}(\rho)
		\middle\Vert
		\mu_{\mathcal B_\theta}(\tau)
		\right),
	\end{aligned}
	\]
	which proves the claim.
\end{proof}

\subsection{Counting accepting projectors}

The preceding sketch shows why low-rank accepting projectors are the objects to
control.  The next lemma estimates how many such projectors can arise from
rank-one adaptive Pauli-measurement trees.

\begin{lemma}
	\label{lem:counting}
	There exists \(C_\revp>0\), depending only on \(\revp\), such that for all
	\(1\le R\le \revK/2\),
	\begin{equation}
	\label{eq:sm-counting-projectors}
	\log\left|\{P\in\mathcal P_{\rm ad}:\rank P\le R\}\right|
	\le C_\revp R\log^2\frac{e\revK}{R}.
	\end{equation}
\end{lemma}

\begin{proof}
	Let
	\[
	P=\sum_{i\in A}\pi_i,
	\qquad |A|=\rank P\le R,
	\]
	{where the \(\pi_i\) are the leaves of a
	rank-one adaptive stabilizer measurement.}
	
	{Retain only the root-to-leaf paths ending at leaves in \(A\).}

	Let \(V_j\) be the set of retained nodes at depth \(j\).
	Since a depth-\(j\) node has \(\revp^j\) possible outcome prefixes,
	\[
	|V_j|\le \revp^j.
	\]
	On the other hand, every retained depth-\(j\) node contains at least one
	accepting rank-one leaf below it, and distinct retained nodes at the same depth
	have disjoint sets of descendant leaves.  Since there are at most \(R\)
	accepting leaves,
	\[
	|V_j|\le R.
	\]
	Therefore
	\[
	|V_j|\le \min(\revp^j,R).
	\]
	
	{ By the above definition, at a retained node of
	depth \(j\), the Pauli labels
	previously measured along that branch span a \(j\)-dimensional isotropic
	subspace
	\[
	S\subseteq \mathbb F_\revp^{2L}.
	\]
	The preceding eigenvalue labels are already contained in the outcome string
	specifying the node.  A new independent Pauli measurement must commute with
	the preceding measurements, so its label lies in \(S^\perp\setminus S\).
	Conditional on the preceding outcomes, labels differing by an element of
	\(S\) induce the same refinement, up to relabeling of the outcomes.  Hence the
	possible next Pauli measurements are overcounted by the nonzero elements of
	\(S^\perp/S\).  Since
	\[
	\dim(S^\perp/S)=2(L-j),
	\]
	their number is at most
	\[
	|S^\perp/S|-1=\revp^{2(L-j)}-1<\revp^{2(L-j)}.
	\]
	After choosing the Pauli measurement, we record the subset of its \(\revp\)
	outcomes whose child subtrees contain at least one accepting leaf.  This gives
	at most \(2^\revp\) choices and also determines which children are retained.
	Both choices may depend on the complete preceding outcome string, so
	branch-dependent adaptivity is included.}
	
	{ To count the resulting pruned trees, order the retained
	nodes at each depth by their outcome strings and pad the list at depth \(j\)
	with dummy symbols to length \(\min(\revp^j,R)\).  Since the retained-outcome
	subset at each actual node determines its retained children, no separate factor
	for the tree topology is required.  The number of such descriptions is at most
	\[
	\prod_{j=0}^{L-1}
	\left(1+2^\revp\revp^{2(L-j)}\right)^{\min(\revp^j,R)}.
	\]
	Different descriptions may yield the same projector, which only overcounts
	\(\mathcal P_{\rm ad}\).  Taking logarithms therefore gives
	\[
	\begin{aligned}
	\log\left|\{P\in\mathcal P_{\rm ad}:\rank P\le R\}\right|
	&\le
	\sum_{j=0}^{L-1}\min(\revp^j,R)
	\log\!\left(1+2^\revp\revp^{2(L-j)}\right)\\
	&\le
	C_\revp\sum_{j=0}^{L-1}\min(\revp^j,R)(L-j+1).
	\end{aligned}
	\]}
	
	To estimate the sum, let
	\(
	J=\lfloor \log_\revp R\rfloor
	\).
	Then \(\revp^J\le R<\revp^{J+1}\).  For \(j\le J\), \(\min(\revp^j,R)=\revp^j\), and writing
	\(m=J-j\) gives
	\[
	\sum_{j=0}^{J}\revp^j(L-j+1)
	=
	\sum_{m=0}^{J}\revp^{J-m}(L-J+m+1)
	\le
	C_\revp R(L-J+1).
	\]
	For \(j>J\), \(\min(\revp^j,R)=R\), so
	\[
	\sum_{j=J+1}^{L-1}R(L-j+1)
	\le
	C_\revp R(L-J+1)^2.
	\]
	Since \(\revK=\revp^L\) and \(\revp^J\le R<\revp^{J+1}\), the quantity \(L-J+1\) is bounded
	above and below by constant multiples, depending only on \(\revp\), of
	\(\log(e\revK/R)\).
	Combining the two estimates yields
	\[
	\sum_{j=0}^{L-1}\min(\revp^j,R)(L-j+1)
	\le
	C_\revp R\log^2\frac{e\revK}{R},
	\]
	after adjusting \(C_\revp\).  This proves the lemma.
\end{proof}

\subsection{Uniform Haar concentration bound for low-rank projectors}
We combine a fixed-projector Haar tail bound with the counting estimate by a
union bound.

\begin{lemma}~\cite[Lemma~1]{sen2005random}
	\label{lem:fixed-tail}
	There is a universal constant \(c_0>0\) with the following property.  Let \(P\)
	be a fixed rank-\(r\) projector on \(\mathbb C^\revK\), and set
	\[
	X_P=\bra\psi P\ket\psi
	\]
	{ for Haar-random \(\ket\psi\)}.  If \(4r/\revK<a\le 1/2\), then
	\begin{equation}
	\label{eq:sm-fixed-projector-tail}
	\Prb[X_P\ge a]\le \exp[-c_0\revK a].
	\end{equation}
\end{lemma}

\begin{proof}
	{
		The case \(r=0\) is trivial, so assume \(r\ge1\).
		Apply Lemma~1 of Ref.~\cite{sen2005random} to the
		\(r\)-dimensional subspace \(\operatorname{ran}P\subset\mathbb C^\revK\),
		whose orthogonal projector is \(P\), with the deviation parameter
		\[
		t:=\frac{\revK a}{r}.
		\]
		Since \(4r/\revK<a\le1/2\), we have
		\[
		4<t\le\frac{\revK}{2r}<\frac{\revK}{r},
		\]
		so the assumptions of the cited lemma are satisfied. Therefore,
		\[
		\Prb[X_P\ge a]
		\le \exp[-c_0tr]
		=\exp[-c_0\revK a]
		\]
		for a universal constant \(c_0>0\).
	}
\end{proof}

\begin{lemma}
	\label{lem:uniform-low-rank}
	There exist \(A_\revp,c_\revp>0\), depending only on \(\revp\), such that with probability
	at least
	\[
	1-\exp[-c_\revp\log^2 \revK],
	\]
	the following holds simultaneously for every \(P\in\mathcal P_{\rm ad}\) with
	\(1\le r=\rank P\le \revK/2\):
	\begin{equation}
	\label{eq:sm-uniform-low-rank-bound}
	\bra\psi P\ket\psi
	\le
	A_\revp\frac r\revK\log^2\frac{e\revK}{r}.
	\end{equation}
\end{lemma}

\begin{proof}
	Write
		\[
		G_R:=R\log^2\frac{e\revK}{R}.
		\]
		Let \(\mathcal R\) be the powers of two not exceeding \(\revK/2\), together
		with the final scale \(\lfloor \revK/2\rfloor\).  For each
		\(R\in\mathcal R\), let
		\(\mathcal E_R\) be the event that some
		\(P\in\mathcal P_{\rm ad}\) with \(\rank P\le R\) satisfies
		\[
		\bra\psi P\ket\psi>A_\revp\frac{G_R}{\revK}.
		\]
		Choose \(A_\revp\) large enough so that
		\[
		\Prb[\mathcal E_R]\le \exp[-c_\revp G_R]
		\]
		for every \(R\in\mathcal R\).
		
		If \(A_\revp G_R/\revK\ge 1\), then \(\mathcal E_R\) is empty.  Suppose
		\(A_\revp G_R/\revK\le 1/2\).  Since \(R\le \revK/2\), we have
		\(\log(e\revK/R)\ge \log(2e)\).  Thus, for \(A_\revp\) large enough,
		\[
		A_\revp\frac{G_R}{\revK}
		=
		A_\revp\frac R\revK\log^2\frac{e\revK}{R}
		\ge \frac{4R}{\revK}
		\ge \frac{4r}{\revK}
		\]
		for every \(r=\rank P\le R\).  Lemmas~\ref{lem:fixed-tail} and
		\ref{lem:counting} then give
		\[
		\Prb[\mathcal E_R]
		\le
		\exp[(C_\revp-c_0A_\revp)G_R]
		\le
		\exp[-c_\revp G_R],
		\]
		for sufficiently large \(A_\revp\).
		
		It remains only to handle
		\[
		\frac12<A_\revp\frac{G_R}{\revK}<1.
		\]
		Then \(\mathcal E_R\) is contained in the event that some \(P\in\mathcal P_{\rm ad}\)
		with \(\rank P\le R\) has \(\bra\psi P\ket\psi>1/2\).  Moreover,
		\[
		\frac R\revK
		<
		\frac{1}{A_\revp\log^2(e\revK/R)}
		\le
		\frac{1}{A_\revp\log^2(2e)}.
		\]
		For sufficiently large \(A_\revp\), \(R/\revK\le 1/8\), and hence
		\(1/2\ge 4r/\revK\) for all \(r\le R\).  Applying Lemma~\ref{lem:fixed-tail} with
		threshold \(1/2\), and then Lemma~\ref{lem:counting}, yields
		\[
		\Prb[\mathcal E_R]\le \exp[C_\revp G_R-c_1\revK]
		\]
		for a universal constant \(c_1>0\).  Since \(A_\revp G_R/\revK<1\), we have
		\(\revK>A_\revp G_R\), so
		\[
		C_\revp G_R-c_1\revK
		\le
		(C_\revp-c_1A_\revp)G_R
		\le
		-c_\revp G_R
		\]
		for sufficiently large \(A_\revp\).  Hence \(\Prb[\mathcal E_R]\le \exp[-c_\revp G_R]\) in all
		cases.
		
		Now take the union bound over \(R\in\mathcal R\).  Since
		\[
		G_R
		=
		R\log^2\frac{e\revK}{R}
		\ge c\log^2 \revK
		\qquad (1\le R\le \revK/2),
		\]
		after adjusting constants for bounded \(\revK\), and since there are only
		\(O(\log \revK)\) scales,
		\[
		\Prb\!\left[\bigcup_{R\in\mathcal R} \mathcal E_R\right]
		\le
		\exp[-c_\revp\log^2 \revK],
		\]
		after decreasing \(c_\revp\).
		
		On the complement of \(\bigcup_{R\in\mathcal R}\mathcal E_R\), take any
		\(P\in\mathcal P_{\rm ad}\) with \(1\le r=\rank P\le \revK/2\).  Choose
		\(R\in\mathcal R\) such that \(r\le R<2r\), which is possible by the
		powers-of-two scales together with the final scale \(\lfloor \revK/2\rfloor\).
		Then
		\[
		\bra\psi P\ket\psi
		\le
		A_\revp\frac R\revK\log^2\frac{e\revK}{R}
		\le
		2A_\revp\frac r\revK\log^2\frac{e\revK}{r}.
		\]
		Absorbing the final factor \(2\) into \(A_\revp\) proves the claim.
	
\end{proof}

\subsection{Bound on stabilizer-visible entanglement for Haar-random states}

\begin{lemma}
	\label{lem:basis-entropy}
	On the good event of Lemma~\ref{lem:uniform-low-rank}, every rank-one adaptive
	stabilizer measurement \(\mathcal B=\{\pi_i\}_{i=1}^\revK\) satisfies
	\begin{equation}
	\label{eq:sm-basis-entropy-bound}
D_{\rm KL}\!\left(
	\mu_{\mathcal B}(\psi)
	\middle\Vert
	\mu_{\mathcal B}({\Id}/\revK)
	\right)
	\le C_\revp .
	\end{equation}
\end{lemma}

\begin{proof}
	Let
	\(
	p_i=\Tr(\pi_i\psi),
	\)
	and order these probabilities so that
	\(
	p_{(1)}\ge p_{(2)}\ge\cdots\ge p_{(\revK)}.
	\)
	Since \(\mu_{\mathcal B}({\Id}/\revK)\) is uniform,
	\begin{equation}
	\label{eq:sm-basis-entropy-identity}
D_{\rm KL}\!\left(
	\mu_{\mathcal B}(\psi)
	\middle\Vert
	\mu_{\mathcal B}({\Id}/\revK)
	\right)
	=
	\sum_{k=1}^\revK p_{(k)}\log(\revK p_{(k)}).
	\end{equation}
	
	For \(1\le k\le \revK/2\), the projector
	\[
	P_k=\sum_{i=1}^k\pi_{(i)}
	\]
	belongs to \(\mathcal P_{\rm ad}\) and has rank \(k\).  Lemma~
	\ref{lem:uniform-low-rank} gives
	\begin{equation}
	\label{eq:sm-cumulative-probability-bound}
	\sum_{i=1}^k p_{(i)}
	=
	\bra\psi P_k\ket\psi
	\le
	A_\revp\frac k\revK\log^2\frac{e\revK}{k}.
	\end{equation}
	Therefore
	\begin{equation}
	\label{eq:sm-pointwise-probability-envelope}
	p_{(k)}
	\le
	\frac{A_\revp}{\revK}\log^2\frac{e\revK}{k}
	=:b_k.
	\end{equation}
	If \(p_{(k)}<1/\revK\), the term \(p_{(k)}\log(\revK p_{(k)})\) is nonpositive.  If
	\(p_{(k)}\ge1/\revK\), then \(x\mapsto x\log(\revK x)\) is increasing on
	\([1/\revK,\infty)\), so
	\[
	p_{(k)}\log(\revK p_{(k)})
	\le
	b_k\log_+(\revK b_k).
	\]
	Consequently, with \(L_k=\log(e\revK/k)\),
	\begin{equation}
	\label{eq:sm-leading-entropy-sum}
	\sum_{k\le \revK/2}p_{(k)}\log(\revK p_{(k)})
	\le
	\frac{A_\revp}{\revK}
	\sum_{k\le \revK/2}
	L_k^2\log_+(A_\revp L_k^2).
	\end{equation}
	The right side is bounded by the
	finite integral
	\[
	\int_0^{1/2}
	\log^2\frac ex\,
	\log_+\!\left(A_\revp\log^2\frac ex\right)\,{\dd}x,
	\]
	whose finiteness follows from the change of variables \(u=\log(e/x)\), since
	\({\dd}x=e^{1-u}{\dd}u\) and the tail is exponentially decaying.
	
	For \(k>\revK/2\), monotonicity gives
	\[
	kp_{(k)}\le \sum_{i=1}^k p_{(i)}\le1,
	\]
	so \(p_{(k)}\le1/k\le2/\revK\).  Hence 
	\(\log(\revK p_{(k)})\leq \log 2\), and
	\[
	\sum_{k>\revK/2}p_{(k)}\log(\revK p_{(k)})
	\le
	\log 2\sum_{k>\revK/2}p_{(k)}
	\le
	\log 2.
	\]
	Combining the two estimates proves the lemma.
\end{proof}

\begin{proof}[Proof of Theorem~\ref{thm:haar-stab-main}]
	The maximally mixed state
	\[
	{\Id}/\revK=({\Id}_{\revp^n}/\revp^n)\otimes({\Id}_{\revp^n}/\revp^n)
	\]
	is separable.  Hence
	\[
	E_{\STAB}(\psi)
	\le
	\sup_{\mathcal M\in\STAB}
D_{\rm KL}\!\left(
	\mu_{\mathcal M}(\psi)
	\middle\Vert
	\mu_{\mathcal M}({\Id}/\revK)
	\right).
	\]
	Work on the good event of Lemma~\ref{lem:uniform-low-rank}.  For any
	\(\mathcal M\in\STAB\), Lemma~\ref{lem:normal-form} gives a convex combination
	of post-processed rank-one adaptive stabilizer measurements.  By the
	data-processing inequality in Lemma~\ref{lem:normal-form} and by
	Lemma~\ref{lem:basis-entropy}, each such measurement contributes at most
	\(C_\revp\).  Therefore
	\[
D_{\rm KL}\!\left(
	\mu_{\mathcal M}(\psi)
	\middle\Vert
	\mu_{\mathcal M}({\Id}/\revK)
	\right)
	\le C_\revp .
	\]
	Combining with the probability estimate in Lemma~\ref{lem:uniform-low-rank} proves the
	theorem.
\end{proof}

\section{Proof of Proposition 2: Werner states}
\label{sec:werner}

{
We recall the Werner parametrization in Eq.~\eqref{eq:werner-state} of the main
text and prove the unrestricted and stabilizer-visible formulas in
Eqs.~\eqref{eq:werner-all} and \eqref{eq:werner-stab}.

Let \(p\) be an odd prime, let \(F_p\) be the swap operator on
\(\mathbb C^p\otimes\mathbb C^p\), and set
\(\Pi_{\pm,p}=({\Id}\pm F_p)/2\).  Define
\begin{equation}
\label{eq:sm-werner-parametrization}
\tau_{\pm,p}=\frac{\Pi_{\pm,p}}{\Tr\Pi_{\pm,p}}
=\frac{{\Id}\pm F_p}{p(p\pm1)},\qquad
\rho_{\mathrm W,p}(t)=t\,\tau_{+,p}+(1-t)\tau_{-,p}.
\end{equation}
Werner states are separable for \(t\ge1/2\), and entangled for
\(0\le t<1/2\) \cite{werner1989quantum,vollbrecht2001entanglement}.

The unrestricted value follows from the known relative entropy of entanglement
of Werner states \cite{vollbrecht2001entanglement}.  In the present
parametrization,
\[
E_R\!\left(\rho_{\mathrm W,p}(t)\right)
=
t\log(2t)+(1-t)\log(2(1-t)),
\qquad 0\le t<1/2,
\]
with the convention \(0\log0=0\).
The same value is attained by unrestricted measurements.  Indeed, for any
separable \(\sigma\), \(\Tr(\Pi_{+,p}\sigma)\ge1/2\), so the two-outcome measurement
\(\{\Pi_{+,p},\Pi_{-,p}\}\) minimizes the separable-side binary distribution at the
boundary parameter \(1/2\).  Hence
	\begin{equation}
	\label{eq:sm-werner-unrestricted-value}
E_{\ALL}\!\left(\rho_{\mathrm W,p}(t)\right)
=E_R\!\left(\rho_{\mathrm W,p}(t)\right)
=
t\log(2t)+(1-t)\log(2(1-t)),
\qquad 0\le t<1/2.
	\end{equation}

\begin{theorem}
	\label{thm:werner}
	For every odd prime \(p\) and \(0\le t<1/2\),
	\begin{equation}
	\label{eq:sm-werner-stabilizer-value}
	E_{\STAB}\!\left(\rho_{\mathrm W,p}(t)\right)
	=
	\frac{2t}{p+1}\log(2t)
	+
	\frac{p+1-2t}{p+1}
	\log\left(\frac{p+1-2t}{p}\right).
	\end{equation}
	Consequently,
	\begin{equation}
	\label{eq:sm-werner-uniform-bound}
	\sup_{0\le t\le1/2}
	E_{\STAB}\!\left(\rho_{\mathrm W,p}(t)\right)
	\le
	\log\left(1+\frac1p\right)
	\end{equation}
	for every odd prime \(p\).
\end{theorem}

\begin{proof}
	Let \(s_+=\Tr\Pi_{+,p}=p(p+1)/2\) and
	\(s_-=\Tr\Pi_{-,p}=p(p-1)/2\).
	{ Let \(\operatorname{Cliff}(p)\) denote the single-qudit
	Clifford group modulo global phases.}  The bilateral
	Clifford twirl
	\(\mathcal T_C(\cdot)=
	\bigl|{\operatorname{Cliff}(p)}\bigr|^{-1}
	\sum_{U\in{\operatorname{Cliff}(p)}}(U\otimes U)(\cdot)
	(U^\dagger\otimes U^\dagger)\) coincides with the
	Haar \(U\otimes U\) twirl in prime dimension by the Clifford unitary
	2-design property \cite{gross2007evenly}.  It therefore fixes Werner states and maps separable states to
	separable Werner states.  Since
	\(D_{\STAB}\) is monotone under bilateral Clifford preprocessing, for every
	separable \(\sigma\),
	\begin{equation}
	\label{eq:sm-werner-twirl-monotonicity}
	D_{\STAB}\!\left(\rho_{\mathrm W,p}(t)\middle\Vert \mathcal T_C(\sigma)\right)
	\le D_{\STAB}\!\left(\rho_{\mathrm W,p}(t)\middle\Vert \sigma\right).
	\end{equation}
	Thus the minimization over separable states may be restricted to Werner
	states \(\rho_{\mathrm W,p}(u)\) with \(u\ge1/2\).
	
	Fix \(u\ge1/2\) and a stabilizer POVM \(\mathcal M=\{E_k\}_k\).  Discard
	outcomes with \(\Tr E_k=0\), and set
	\(r_k=\Tr E_k\) and \(c_k=\Tr(E_k\Pi_{+,p})/r_k\).  Since
	\(0\le\Pi_{+,p}\le{\Id}\), \(c_k\le1\).  Also
	\(c_k=\frac12(1+\Tr(E_kF_p)/\Tr E_k)\).  In odd prime dimension,
	\begin{equation}
	\label{eq:sm-swap-wigner-expansion}
	F_p=\frac1p\sum_z A_z\otimes A_z,\qquad
	\Tr(E_kF_p)=\frac1p\sum_z W_{E_k}(z,z)\ge0.
	\end{equation}
	Hence \(c_k\in[1/2,1]\).
	
	For \(x\in[0,1]\), the outcome probability on \(\rho_{\mathrm W,p}(x)\) is
	\begin{equation}
	\label{eq:sm-werner-outcome-probability}
	\Tr\!\left(E_k\rho_{\mathrm W,p}(x)\right)=r_kB_x(c_k),\qquad
	B_x(c)=\frac{x c}{s_+}+\frac{(1-x)(1-c)}{s_-}.
	\end{equation}
	With \(h(c)=B_t(c)\log[B_t(c)/B_u(c)]\), this gives
	\[
	D_{\rm KL}\!\left(\mu_{\mathcal M}(\rho_{\mathrm W,p}(t))\Vert
	\mu_{\mathcal M}(\rho_{\mathrm W,p}(u))\right)=\sum_k r_k h(c_k).
	\]
	A direct differentiation gives
	\[
	h''(c)=(\log e)B_t(c)
	\left(\frac{B_t'(c)}{B_t(c)}-\frac{B_u'(c)}{B_u(c)}\right)^2\ge0.
	\]
	Thus \(h\) is convex on \([1/2,1]\), and
	\(h(c)\le(2c-1)h(1)+2(1-c)h(1/2)\) on this interval.  Using
	\(\sum_kE_k={\Id}\), one obtains
	\[
	\sum_k r_k(2c_k-1)=\Tr F_p=p,\qquad
	\sum_k 2r_k(1-c_k)=2\Tr\Pi_{-,p}=p(p-1).
	\]
	Consequently every stabilizer POVM satisfies
	\begin{equation}
	\label{eq:sm-werner-measurement-upper-bound}
	D_{\rm KL}\!\left(\mu_{\mathcal M}(\rho_{\mathrm W,p}(t))\Vert
	\mu_{\mathcal M}(\rho_{\mathrm W,p}(u))\right)
	\le p h(1)+p(p-1)h(1/2).
	\end{equation}
	
	This upper bound is attained by the computational-basis product measurement
	\(\{|i,j\rangle\!\langle i,j|\}_{i,j=1}^p\), which is a stabilizer
	measurement.  For \(E_{ij}=|i,j\rangle\!\langle i,j|\), one has
	\(r_{ij}=1\) and \(c_{ij}=(1+\delta_{ij})/2\).  Hence \(c=1\) on the \(p\)
	diagonal outcomes and \(c=1/2\) on the \(p(p-1)\) off-diagonal outcomes, so
	Eq.~\eqref{eq:sm-werner-measurement-upper-bound} is tight.  Therefore
	\begin{equation}
	\label{eq:sm-werner-stabilizer-relative-entropy}
	D_{\STAB}\!\left(\rho_{\mathrm W,p}(t)\Vert\rho_{\mathrm W,p}(u)\right)
	=
	\frac{2t}{p+1}\log\frac{t}{u}
	+
	\frac{p+1-2t}{p+1}
	\log\frac{p+1-2t}{p+1-2u}.
	\end{equation}
	Its derivative in \(u\) is
	\(2(\log e)(u-t)/[u(p+1-2u)]>0\) for \(u\ge1/2\) and
	\(t<1/2\).  The
	minimum over separable Werner states is therefore attained at \(u=1/2\),
	which gives the stated formula for
	\(E_{\STAB}(\rho_{\mathrm W,p}(t))\).
	
	For the uniform bound, let \(f(t)\) denote the displayed formula.  For
	\(0<t\le1/2\),
	\[
	f'(t)=\frac{2}{p+1}\log\frac{2pt}{p+1-2t}\le0.
	\]
	Thus \(f\) is nonincreasing on \([0,1/2]\). We get
	\(E_{\STAB}(\rho_{\mathrm W,p}(t))\le f(0)=\log(1+1/p)\).
\end{proof}
}

\section{Proofs of Theorems 3 and 4: mixed-stabilizer visibility gaps}
\label{sec:mixed}

{This section proves the two-qutrit visibility gap and the
asymptotic visibility separation stated in the main text.
We also analyze an interpolating family extending the
two-qutrit example.}

The two-qutrit example is the mixed stabilizer state
\begin{equation}
\label{eq:sm-rho-star-construction}
\rho_*=\frac13\ket{00}\bra{00}
+\frac23\ket{\Psi^+_{12}}\bra{\Psi^+_{12}},
\qquad
\ket{\Psi^+_{12}}=\frac{\ket{11}+\ket{22}}{\sqrt2}.
\end{equation}
It has the one-shot gap
\[
E_{\ALL}(\rho_*)=\frac23,
\qquad
E_{\STAB}(\rho_*)\le \frac13\log 3.
\]
The amplified construction uses larger qutrit stabilizer blocks \(\rho_b\), for
which
\[
E_{\ALL}(\rho_b)=1-3^{-b},
\qquad
E_{\STAB}(\rho_b)\le \frac{b\log3}{3^b}.
\]
Choosing \(b_N=\lceil2\log_3 N\rceil\) and taking many independent blocks gives
the asymptotic separation of Theorem~4.

{ Throughout this section, \(\Delta\) denotes the complete dephasing channel in the local computational bases.}

Both arguments use the dephased separable reference
\(\sigma=\Delta(\rho)\).  For the maximally correlated states below this
reference gives the unrestricted value, while the stabilizer side is bounded by
\[
E_{\STAB}(\rho)
\le D_{\STAB}(\rho\Vert\sigma).
\]
The states used below are PWF, and stabilizer measurements act on their Wigner
functions through a classical stochastic kernel.  Classical data processing
therefore yields the computable Wigner--KL upper bound
\[
D_{\STAB}(\rho\Vert\sigma)
\le D_{\rm KL}(W_\rho\Vert W_\sigma).
\]
\subsection{Wigner--KL upper bounds for positive-Wigner states}

The following lemma establishes this bound.

\begin{lemma}[Wigner--KL upper bound for PWF states]
\label{lem:pwf-dpi}
Let \(\rho\) and \(\sigma\) be PWF states on the same system.  Then
\begin{equation}
\label{eq:sm-pwf-data-processing}
D_{\STAB}(\rho\Vert\sigma)
\le D_{\rm KL}(W_\rho\Vert W_\sigma).
\end{equation}
\end{lemma}

\begin{proof}
For a stabilizer POVM \(\{E_k\}\), each effect is PWF, so
\(W_{E_k}(u)\ge0\).  Completeness gives
\[
\sum_k W_{E_k}(u)=W_{{\Id}}(u)=\Tr A_u=1,
\]
so \(\Gamma_{k|u}:=W_{E_k}(u)\) is a stochastic kernel from phase-space points
to outcomes.  By Eq.~\eqref{eq:sm-wigner-born-rule},
\[
\mu_{\mathcal M}(\rho)(k)=\Tr(E_k\rho)
=\sum_u W_{E_k}(u)W_\rho(u)
=\sum_u\Gamma_{k|u}W_\rho(u),
\]
and the same channel maps \(W_\sigma\) to \(\mu_{\mathcal M}(\sigma)\).
Classical data processing gives
for every \(\mathcal M\in\STAB\),
\begin{equation*}
D_{\rm KL}(\mu_{\mathcal M}(\rho)\Vert\mu_{\mathcal M}(\sigma))
\le D_{\rm KL}(W_\rho\Vert W_\sigma).
\end{equation*}
Taking the supremum over stabilizer measurements proves the claim.
\end{proof}

\begin{lemma}[Regularized Wigner--KL upper bound for PWF states]
\label{lem:pwf-regularized}
Let \(\rho\) be a PWF state on a bipartite tensor-product qutrit system and let
\(\sigma\) be a PWF separable state on the same system.  Then
\begin{equation}
\label{eq:sm-regularized-pwf-upper-bound}
\limsup_{M\to\infty}\frac1M
E_{\STAB}(\rho^{\otimes M})
\le D_{\rm KL}(W_\rho\Vert W_\sigma).
\end{equation}
\end{lemma}

\begin{proof}
For every \(M\ge1\), the product reference \(\sigma^{\otimes M}\) is separable
and PWF.  Lemma~\ref{lem:pwf-dpi} gives
\[
E_{\STAB}(\rho^{\otimes M})
\le
D_{\STAB}\!\left(\rho^{\otimes M}\Vert\sigma^{\otimes M}\right)
\le
D_{\rm KL}\!\left(W_{\rho}^{\otimes M}\Vert W_{\sigma}^{\otimes M}\right)
=M D_{\rm KL}(W_\rho\Vert W_\sigma).
\]
The equality follows from additivity of classical KL divergence for product
distributions.  Dividing by \(M\) and
taking the limsup proves the claim.
\end{proof}

\subsection{A two-qutrit mixed-stabilizer gap}

\begin{theorem}
\label{thm:rho-star-gap}
Let \(\rho_*\) be the state defined in
Eq.~\eqref{eq:sm-rho-star-construction}.  Then \(\rho_*\) is a mixed
stabilizer state and
\begin{equation}
\label{eq:sm-rho-star-gap}
E_{\STAB}(\rho_*)\le \frac13\log 3<E_{\ALL}(\rho_*)=\frac23.
\end{equation}
\end{theorem}

\begin{proof}
Let
\begin{equation}
\label{eq:sm-rho-star-reference}
\sigma_*=\Delta(\rho_*)=\frac13\sum_{i=0}^2\ket{ii}\bra{ii}.
\end{equation}
First, \(\rho_*\) is { a mixed stabilizer state}.  Let
\(\omega=e^{2\pi i/3}\), and define
\begin{equation*}
\ket{{\xi}_a}=\frac1{\sqrt3}\sum_{j=0}^{2}\omega^{a j^2}\ket{jj},
\qquad a=0,1,2.
\end{equation*}
Each \(\ket{{\xi}_a}\) is a two-qutrit stabilizer state: with
\(X\ket j=\ket{j+1}\) and \(Z\ket j=\omega^j\ket j\), it is stabilized by
\begin{equation*}
Z_AZ_B^{-1},
\qquad
\omega^{-a}(Z_A^aX_A)\otimes(Z_B^aX_B).
\end{equation*}
Moreover,
\begin{equation}
\label{eq:sm-rho-star-stabilizer-mixture}
\begin{aligned}
\frac13\sum_{a=0}^2\ket{{\xi}_a}\bra{{\xi}_a}
&=\frac13\sum_{j,k=0}^{2}
\left(\frac13\sum_{a=0}^{2}\omega^{a(j^2-k^2)}\right)
\ket{jj}\bra{kk}\\
&=\frac13\ket{00}\bra{00}
+\frac23\ket{\Psi^+_{12}}\bra{\Psi^+_{12}}
=\rho_*.
\end{aligned}
\end{equation}
Thus \(\rho_*\) is a convex mixture of stabilizer states.

Next, \(\rho_*\) is maximally correlated, so the dephased state
\(\sigma_*=\Delta(\rho_*)\) gives the relative entropy of entanglement
{\cite{rains1999bound,rains2000erratum,zhu2017coherence}}.  The unrestricted quantity reaches the same
value { by using a projective measurement containing the nonzero eigenvectors}
\(\{\ket{\varphi_j}\}=\{\ket{00},\ket{\Psi^+_{12}}\}\) { of \(\rho_*\)}.  For a product vector
\(\ket a\ket b\),
\[
\bra{a,b}\rho_*\ket{a,b}
=\frac13\left(|a_0b_0|^2+|a_1b_1+a_2b_2|^2\right)
\le\frac13,
\]
by Cauchy--Schwarz, and convexity gives the same bound for every separable
state.  Hence the
classical reference probabilities
\(r_j=\langle\varphi_j|\eta|\varphi_j\rangle\) in this eigenbasis obey
\(\sum_j p_jr_j=\Tr(\rho_*\eta)\le1/3\), where
\(\eta\) is separable and \(p=(1/3,2/3)\).  Jensen's inequality gives
\[
D_{\rm KL}(p\Vert r)
\ge
\sum_j p_j\log(3p_j)
=D(\rho_*\Vert\sigma_*).
\]
Together with \(E_{\ALL}\le E_R\), this proves
\begin{equation}
\label{eq:sm-rho-star-unrestricted-value}
E_{\ALL}(\rho_*)=D(\rho_*\Vert\sigma_*).
\end{equation}
On the correlated subspace, \(\sigma_*\) is \({\Id}_3/3\), while \(\rho_*\) has
nonzero eigenvalues \(1/3\) and \(2/3\).  Therefore
\begin{equation*}
D(\rho_*\Vert\sigma_*)=\frac13\log1+\frac23\log2=\frac23.
\end{equation*}

For the stabilizer upper bound, since \(\sigma_*\) is separable,
\begin{equation*}
E_{\STAB}(\rho_*)\le D_{\STAB}(\rho_*\Vert\sigma_*).
\end{equation*}
Both \(\rho_*\) and \(\sigma_*\) are mixtures of stabilizer states, hence have
{ nonnegative} Gross Wigner functions.  Lemma~\ref{lem:pwf-dpi}, applied directly
to the full two-qutrit Wigner distributions, gives
\begin{equation*}
D_{\STAB}(\rho_*\Vert\sigma_*)
\le D_{\rm KL}(W_{\rho_*}\Vert W_{\sigma_*}).
\end{equation*}
A direct evaluation of these two Wigner distributions gives
\[
W_{\sigma_*}=1/27
\]
on \(27\) phase-space points and zero elsewhere, while \(W_{\rho_*}\) has weight
\(1/9\) on \(3\) of those points, weight \(1/27\) on \(18\) of those points, and
zero elsewhere.  Therefore
\begin{equation}
\label{eq:sm-rho-star-wigner-kl}
D_{\rm KL}(W_{\rho_*}\Vert W_{\sigma_*})
=3\cdot\frac19\log3
=\frac13\log3
<\frac23.
\end{equation}
\end{proof}

{The same argument extends to the interpolating family
\begin{equation}
\label{eq:rho-interpolation}
\rho(s)=(1-s)\sigma_*+s\rho_*,
\qquad 0\le s\le1,
\end{equation}
which connects the separable reference \(\sigma_*\) to the
mixed stabilizer state \(\rho_*\).
Figure~\ref{fig:mixed-stabilizer-gap} illustrates the
visibility gap along this family.}

\begin{figure}[htbp]
\centering
    \includegraphics[width=0.55\linewidth]{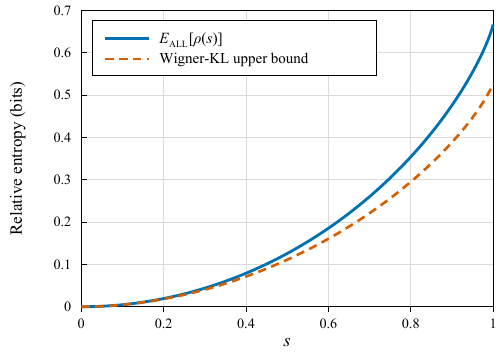}
\caption{Mixed-stabilizer visibility gap. For \(\rho(s)=(1-s)\sigma_*+s\rho_*\), the solid curve gives the exact unrestricted value \(E_{\ALL}(\rho(s))=D(\rho(s)\Vert\sigma_*)\).  { The dashed curve is \(D_{\rm KL}(W_{\rho(s)}\Vert W_{\sigma_*})\), the Wigner--KL upper bound on \(E_{\STAB}(\rho(s))\).  Since this certified upper bound is strictly smaller than \(E_{\ALL}(\rho(s))\) for every \(s>0\), stabilizer measurements are strictly weaker than unrestricted measurements} for this mixed-stabilizer family.}
\label{fig:mixed-stabilizer-gap}
\end{figure}

By the maximally correlated formula, the solid blue curve is
the exact unrestricted value
\begin{equation}
\label{eq:sm-interpolation-unrestricted-value}
E_{\ALL}(\rho(s))=D(\rho(s)\Vert\sigma_*)
=\frac{1+s}{3}\log(1+s)+\frac{1-s}{3}\log(1-s).
\end{equation}
The dashed orange curve is the Wigner--KL upper bound
\(U_W(s):=D_{\rm KL}(W_{\rho(s)}\Vert W_{\sigma_*})\), obtained from
Lemma~\ref{lem:pwf-dpi}:
\begin{equation}
\label{eq:sm-interpolation-wigner-bound}
E_{\STAB}(\rho(s))
\le
U_W(s)
=
\frac{1+2s}{9}\log(1+2s)
+\frac{2(1-s)}{9}\log(1-s).
\end{equation}
Thus the orange dashed curve should be read as a certified upper bound on the
stabilizer-visible value, not as the exact value of \(E_{\STAB}\).
{For every \(0<s\le1\), this upper bound is strictly smaller
than the unrestricted value, certifying a strict
visibility gap throughout the nontrivial part of the interpolation.}

\subsection{A family of mixed-stabilizer hiding states}

Let \(b\ge1\), set \(d=3^b\), and fix an
\(\mathbb F_3\)-linear identification of the computational-basis
labels of \(b\) qutrits with \(\mathbb F_d\).
Let
\(\Tr_d(t):=\operatorname{Tr}_{\mathbb F_d/\mathbb F_3}(t)
=\sum_{j=0}^{b-1}t^{3^j}\).
The canonical additive character is
\begin{equation*}
	\label{eq:sm-additive-character}
	\chi(t):=
	\exp\!\left(\frac{2\pi i}{3}\Tr_d(t)\right),
	\qquad t\in\mathbb F_d .
\end{equation*}
It satisfies the standard orthogonality relation
\[
\sum_{a\in\mathbb F_d}\chi(at)=d\,\delta_{t,0}.
\]
Following the standard finite-field construction of mutually unbiased bases
in odd prime-power dimension
\cite{wootters1989optimal,gibbons2004discrete}, define
{
\begin{equation}
\label{eq:sm-quadratic-phase-stabilizers}
	\ket{\zeta_a}=\frac1{\sqrt d}\sum_{x\in\mathbb F_d}\chi(ax^2)\ket x,
	\qquad a\in\mathbb F_d .
\end{equation}
}
Let 
\[
X(t)\ket x=\ket{x+t},\qquad Z(s)\ket x=\chi(sx)\ket x .
\]
Set
\[
L_a=\{(x,2ax):x\in\mathbb F_d\}\subset\mathbb F_3^b\oplus\mathbb F_3^b .
\]
\begin{samepage}
{ We use the following properties of \(\ket{\zeta_a}\):}
\begin{enumerate}[label=(\roman*)]
	\item { \(\ket{\zeta_a}\) is a stabilizer state.}  Indeed, for every
	\(t\in\mathbb F_d\),
	{
	\[
	\chi(-at^2)Z(2at)X(t)\ket{\zeta_a}=\ket{\zeta_a},
	\]
	}
	and these commuting Paulis form the stabilizer associated with the line
	\(L_a\) \cite{gibbons2004discrete,gross2006hudson}.
	\item { \(\ket{\zeta_a}\) is flat in the computational basis}
	\cite{wootters1989optimal,gibbons2004discrete}:
	{
	\[
	|\braket{x}{\zeta_a}|=d^{-1/2}.
	\]
	}
	\item By Lemma 9 of Ref.~\cite{gross2006hudson}, the tensor-product
	qutrit Wigner function of
	{ \(\ket{\zeta_a}\) is uniform on \(L_a\).}
	\item The supports satisfy \(L_a\cap L_{a'}=\{0\}\) whenever \(a\ne a'\).
\end{enumerate}
\end{samepage}
Define the correlated states
{
\begin{equation}
\label{eq:sm-correlated-block-construction}
\ket{{\xi}_a^{(b)}}
=\frac1{\sqrt d}\sum_{x\in\mathbb F_d}\chi(ax^2)\ket{x,x},
\qquad
\rho_b=\frac1d\sum_{a\in\mathbb F_d}
\ket{{\xi}_a^{(b)}}\bra{{\xi}_a^{(b)}}.
\end{equation}
}
{ The repetition isometry \(V_{\rm rep}\ket x=\ket{x,x}\) is a
stabilizer isometry, so each \(\ket{{\xi}_a^{(b)}}\) is a stabilizer
state} and \(\rho_b\) is { a mixed stabilizer state}.
It is maximally correlated by construction.  { For \(b=1\),
\(\ket{{\xi}_a^{(1)}}=\ket{{\xi}_a}\), so this construction
reduces to the state \(\rho_*\) above.}  Let
\begin{equation}
\label{eq:sm-block-separable-reference}
\sigma_b=\Delta(\rho_b)=\frac1d\sum_{x\in\mathbb F_d}\ket{x,x}\bra{x,x}.
\end{equation}
Define
\begin{equation}
	\label{eq:sm-block-seed-state}
	\gamma_b:=\frac1d\sum_{a\in\mathbb F_d}
	\ket{\zeta_a}\bra{\zeta_a}.
\end{equation}

\begin{lemma}
\label{lem:block-all}
For the above \(\rho_b\),
\begin{equation}
\label{eq:sm-block-unrestricted-values}
S(\rho_b)=\log d-\left(1-\frac1d\right),
\qquad
E_{\ALL}(\rho_b)=D(\rho_b\Vert\sigma_b)=1-1/d.
\end{equation}
\end{lemma}

\begin{proof}
	{ By construction, \(\rho_b=V_{\rm rep}\gamma_bV_{\rm rep}^\dagger\) and
	\(\sigma_b=V_{\rm rep}({\Id}_d/d)V_{\rm rep}^\dagger\).}  By character orthogonality,
	{
	\begin{equation}
	\label{eq:sm-block-seed-matrix-elements}
	\bra x\gamma_b\ket y
	=\frac1{d^2}\sum_{a\in\mathbb F_d}\chi\!\left(a(x^2-y^2)\right)
	=
	\begin{cases}
	1/d, & x^2=y^2,\\
	0, & x^2\ne y^2 .
	\end{cases}
	\end{equation}
	}
	Since the field has odd characteristic, \(x^2=y^2\) means \(x=y\) or \(x=-y\).
	Thus \(\ket0\) has eigenvalue \(1/d\), while each two-dimensional space
	\(\operatorname{span}\{\ket x,\ket{-x}\}\), \(x\ne0\), contributes
	eigenvalues \(2/d\) and \(0\).  Hence
	{
	\[
	S(\rho_b)=S(\gamma_b)=\log d-\left(1-\frac1d\right),
	\qquad
	D(\rho_b\Vert\sigma_b)=\log d-S(\rho_b)=1-\frac1d.
	\]
	}
	For { maximally correlated states}, the dephased state
	\(\sigma_b=\Delta(\rho_b)\) attains the relative entropy of entanglement
	{\cite{rains1999bound,rains2000erratum,zhu2017coherence}}.  For the present family this gives
	\[
	E_R(\rho_b)=D(\rho_b\Vert\sigma_b)=1-\frac1d.
	\]
	We now show that \(E_{\ALL}(\rho_b)=E_R(\rho_b)\).
	Let \(\rho_b=\sum_j p_j|\varphi_j\rangle\langle\varphi_j|\), and measure in this
	eigenbasis.  { For any separable \(\eta\), with
	\(q_j=\langle\varphi_j|\eta|\varphi_j\rangle\), one has}
	{
	\[
	\sum_jp_jq_j=\Tr(\rho_b\eta)\le\frac1d,
	\]
	}
	{ because each \(\ket{{\xi}_a^{(b)}}\) is maximally entangled} and therefore has
	overlap at most \(1/d\) with any separable state.  Jensen's inequality gives
	\[
	D_{\rm KL}(p\Vert q)
	\ge \sum_jp_j\log(dp_j)
	= D(\rho_b\Vert\sigma_b).
	\]
	Thus \(E_{\ALL}(\rho_b)\ge D(\rho_b\Vert\sigma_b)\), while
	\(E_{\ALL}(\rho_b)\le E_R(\rho_b)\).  Hence
	\[
	E_{\ALL}(\rho_b)=1-\frac1d.
	\]
\end{proof}

\begin{lemma}
\label{lem:block-stab}
For the above \(\rho_b\),
\begin{equation}
\label{eq:sm-block-stabilizer-visible}
E_{\STAB}(\rho_b)
\le \frac{\log d}{d}
= \frac{b\log 3}{3^b}.
\end{equation}
\end{lemma}

\begin{proof}
Since \(\sigma_b\) is separable and both states are PWF,
Lemma~\ref{lem:pwf-dpi} gives the Wigner--KL upper bound
\begin{equation*}
E_{\STAB}(\rho_b)
\le D_{\rm KL}(W_{\rho_b}\Vert W_{\sigma_b}).
\end{equation*}
To evaluate this KL divergence, let \(U\) be a Clifford
decoder for the repetition code, chosen so that
\[
U\ket{x,x}=\ket{x}\ket{0^b}.
\]
Then
{
\begin{equation}
\label{eq:sm-decoded-block-states}
U\rho_bU^\dagger=\gamma_b\otimes\ket{0^b}\bra{0^b},
\qquad
U\sigma_bU^\dagger=({\Id}_d/d)\otimes\ket{0^b}\bra{0^b}.
\end{equation}
}
By Clifford covariance of the Wigner representation, conjugation by \(U\) only
permutes phase-space points and therefore leaves classical KL divergence
unchanged.  Since the second tensor factor is identical in the two decoded
states, additivity of KL divergence for product distributions gives
{
\begin{equation*}
D_{\rm KL}(W_{\rho_b}\Vert W_{\sigma_b})
=D_{\rm KL}(W_{\gamma_b}\Vert W_{{\Id}_d/d}).
\end{equation*}
}

The maximally mixed state \({\Id}_d/d\) has a uniform Wigner
distribution on \(d^2=3^{2b}\) phase-space points.
The Wigner function of \(\gamma_b\) is the average of the uniform
distributions on the Lagrangian subspaces \(L_a\).
Since these subspaces intersect only at the origin,
{
\begin{equation}
\label{eq:sm-block-seed-wigner-function}
W_{\gamma_b}(0)=1/d,\qquad
W_{\gamma_b}(u)=1/d^2
\end{equation}
}
for \(u\ne0\) lying in one of the subspaces \(L_a\), and { \(W_{\gamma_b}(u)=0\)}
elsewhere.  All nonzero supported points have the same weight as the uniform
reference, so only the origin contributes to the KL divergence:
{
\begin{equation}
\label{eq:sm-block-wigner-kl}
D_{\rm KL}(W_{\gamma_b}\Vert W_{{\Id}_d/d})
=\frac1d\log d.
\end{equation}
}
\end{proof}

\begin{theorem}[Magic-free asymptotic hiding of visible entanglement]
\label{thm:amplified}
For \(N\ge2\), set
\begin{equation}
\label{eq:sm-amplified-state-construction}
\begin{gathered}
b_N=\lceil2\log_3 N\rceil,\qquad
r_N=\lfloor N/b_N\rfloor,\\
\widehat\rho_N=
\rho_{b_N}^{\otimes r_N}
\otimes(\ket{00}\bra{00})^{\otimes(N-b_Nr_N)}.
\end{gathered}
\end{equation}
Then \(\widehat\rho_N\) is a mixed stabilizer state on \(N\) qutrits per party,
and
\begin{equation}
\label{eq:sm-amplified-visibility-separation}
\begin{aligned}
E_{\ALL}(\widehat\rho_N)&=\Omega(N/\log N),\\
E_{\STAB}(\widehat\rho_N)&\le \frac{\log 3}{N}.
\end{aligned}
\end{equation}
\end{theorem}

\begin{proof}
The argument of Lemma~\ref{lem:block-all} tensorizes. Indeed,
\(\rho_{b_N}^{\otimes r_N}\) remains maximally correlated and is a
mixture of maximally entangled states of local dimension
\(3^{b_Nr_N}\). Hence its overlap with every separable state is at
most \(3^{-b_Nr_N}\), and the same eigenbasis-measurement and Jensen
argument gives the matching lower bound. Together with
\(E_{\ALL}\le E_R\), this yields
\begin{equation}
	\label{eq:sm-amplified-unrestricted-product}
	E_{\ALL}(\rho_{b_N}^{\otimes r_N})
	=D(\rho_{b_N}^{\otimes r_N}\Vert
	\sigma_{b_N}^{\otimes r_N})
	=r_N(1-3^{-b_N}).
\end{equation}

For the stabilizer-visible upper bound, use the product separable reference
\(\sigma_{b_N}^{\otimes r_N}\) together with the product padding.  Lemma~\ref{lem:pwf-dpi}
and additivity of classical KL divergence give
\begin{equation}
\label{eq:sm-amplified-stabilizer-visible-bound}
E_{\STAB}(\widehat\rho_N)
\le r_N\frac{b_N\log 3}{3^{b_N}}
\le \frac{\log 3}{N},
\end{equation}
because \(r_Nb_N\le N\) and \(3^{b_N}\ge N^2\).
\end{proof}

\section{Proofs of Corollary 5 and Theorem 6: stabilizer entanglement distillation}
\label{sec:distill}

{ We prove the Haar one-shot bound
\eqref{eq:haar-distill-bound}, the asymptotic converse
\eqref{eq:asymp-converse-main}, and the mixed-state separation
\eqref{eq:distill-separation} of the main text.}

The first elementary fact is the monotonicity of stabilizer-visible entanglement under \(\LSCC\).
\begin{lemma}
\label{lem:lscc-monotone}
For every \(\LSCC\) channel \(\Lambda\),
\begin{equation}
\label{eq:sm-lscc-monotonicity}
E_{\STAB}(\Lambda(\rho))
\le E_{\STAB}(\rho).
\end{equation}
\end{lemma}

\begin{proof}
An \(\LSCC\) channel maps separable states to separable states.  If
\(\mathcal M\in\STAB\) is a stabilizer measurement on the output, then
\(\mathcal M\circ\Lambda\) is a stabilizer measurement on the input.  Hence, for
every \(\sigma\in\SEP\),
\begin{equation*}
D_{\STAB}(\Lambda(\rho)\Vert\Lambda(\sigma))
\le D_{\STAB}(\rho\Vert\sigma).
\end{equation*}
Therefore
\begin{align*}
E_{\STAB}(\Lambda(\rho))
&=\inf_{\tau\in\SEP}D_{\STAB}(\Lambda(\rho)\Vert\tau)\\
&\le \inf_{\sigma\in\SEP}
D_{\STAB}(\Lambda(\rho)\Vert\Lambda(\sigma))\\
&\le \inf_{\sigma\in\SEP}D_{\STAB}(\rho\Vert\sigma)
=E_{\STAB}(\rho),
\end{align*}
where the first inequality uses \(\Lambda(\sigma)\in\SEP\).
\end{proof}

{ The second elementary fact establishes the normalization of
\(E_{\STAB}\) on \(\Phi_p^{\otimes m}\) and provides a robust lower bound for
states close to this target.}
\begin{lemma}
\label{lem:bell-continuity}
{ For every \(m\ge0\),}
\begin{equation}
\label{eq:sm-maximally-entangled-normalization}
E_{\STAB}({\Phi_p^{\otimes m}})=m\log p.
\end{equation}
More generally, if \(\frac12\|\tau-{\Phi_p^{\otimes m}}\|_1\le\epsilon<1\), then
\begin{equation}
\label{eq:sm-maximally-entangled-robust-bound}
E_{\STAB}(\tau)\ge (1-\epsilon)m\log p-\log 2.
\end{equation}
\end{lemma}

\begin{proof}
{ The two-outcome stabilizer test} \(\{{\Phi_p^{\otimes m}},{\Id}-{\Phi_p^{\otimes m}}\}\) accepts
\({\Phi_p^{\otimes m}}\) with probability one, while every separable state has overlap at
most \(p^{-m}\) with \({\Phi_p^{\otimes m}}\).  Thus
\(E_{\STAB}({\Phi_p^{\otimes m}})\ge m\log p\).  The reverse inequality follows from
\(E_{\STAB}\le E_{\ALL}\le E_R\) and the standard separable reference giving
\(E_R({\Phi_p^{\otimes m}})=m\log p\).

If \(\frac12\|\tau-{\Phi_p^{\otimes m}}\|_1\le\epsilon\), the same test accepts \(\tau\)
with probability \(\alpha\ge1-\epsilon\).  For any separable \(\sigma\), its
acceptance probability \(\beta\) satisfies \(\beta\le p^{-m}\).  Hence
\begin{equation*}
D_{\rm KL}((\alpha,1-\alpha)\Vert(\beta,1-\beta))
\ge \alpha\log\frac1\beta-h_2(\alpha)
\ge (1-\epsilon)m\log p-\log 2,
\end{equation*}
where \(h_2\) is the binary entropy and the last inequality uses
\(h_2(\alpha)\le\log 2\).  This gives the claimed continuity lower bound.
\end{proof}

\begin{corollary}[Haar one-shot distillation]
{
For every fixed \(0<\epsilon<1\), there exist constants
\(C_{p,\epsilon},c_p>0\) such that, for Haar-random \(\ket\psi\) on
\(\mathbb C^{p^n}\otimes\mathbb C^{p^n}\), with
\(\psi:=\ket\psi\!\bra\psi\),
\begin{equation}
\label{eq:sm-haar-lscc-distillation}
\Prb\!\left[E_{\LSCC}^{D,(1),\epsilon}(\psi)\le C_{p,\epsilon}\right]
\ge 1-e^{-c_p n^2}.
\end{equation}
}
\end{corollary}

\begin{proof}
{
Work on the event in Theorem~\ref{thm:haar-stab-main}, on which
\(E_{\STAB}(\psi)\le C_p\).  Suppose that an \(\LSCC\) protocol produces a
state \(\tau\) satisfying
\(\frac12\|\tau-\Phi_p^{\otimes m}\|_1\le\epsilon\).  Lemmas~\ref{lem:lscc-monotone}
and \ref{lem:bell-continuity} give
\[
C_p\ge E_{\STAB}(\psi)\ge E_{\STAB}(\tau)
\ge (1-\epsilon)m\log p-\log 2.
\]
Since \(p\) and \(\epsilon\) are fixed, this bounds \(m\log p\) by a constant
\(C_{p,\epsilon}\).  Taking the supremum over achievable \(m\) proves the
claim.  The probability estimate follows from Theorem~\ref{thm:haar-stab-main},
using \(\revK=p^{2n}\) and absorbing the fixed factor \((2\log p)^2\) into
\(c_p\).
}
\end{proof}

\begin{lemma}[Asymptotic relative-entropy converse]
\label{lem:asymp-converse}
For any state \(\rho\),
\begin{equation}
\label{eq:sm-asymptotic-relative-entropy-converse}
E_{\LSCC}^{D}(\rho)
\le \limsup_{M\to\infty}\frac1M E_{\STAB}(\rho^{\otimes M}).
\end{equation}
\end{lemma}

\begin{proof}
Let a protocol distill { \(m_M\) copies of \(\Phi_p\)} from
\(\rho^{\otimes M}\) with error \(\epsilon<1/2\), and write
{ \(\rho_M^{\rm out}=\Lambda_M(\rho^{\otimes M})\).}  Lemma~\ref{lem:lscc-monotone} and
Lemma~\ref{lem:bell-continuity} give
{
\begin{equation}
\label{eq:sm-distillation-converse-chain}
E_{\STAB}(\rho^{\otimes M})
\ge E_{\STAB}(\rho_M^{\rm out})
\ge (1-\epsilon)m_M\log p-\log 2.
\end{equation}
}
Thus
\begin{equation*}
\frac{m_M\log p}{M}\le
\frac{1}{1-\epsilon}
\left(
\frac1M E_{\STAB}(\rho^{\otimes M})
\;+\;
\frac{\log 2}{M}
\right).
\end{equation*}
Taking the limsup over \(M\) and then \(\epsilon\to0\) proves the claim.
\end{proof}

\begin{theorem}[Magic-free asymptotic hiding of distillable entanglement]
For the { mixed stabilizer states} \(\widehat\rho_N\) of Theorem~\ref{thm:amplified},
\begin{equation}
\label{eq:sm-distillation-separation}
E_{\LOCC}^{D}(\widehat\rho_N)=\Omega(N/\log N),
\qquad
E_{\LSCC}^{D}(\widehat\rho_N)\le \frac{\log 3}{N}.
\end{equation}
\end{theorem}

\begin{proof}
For a single block, Bob's marginal is { maximally mixed on a \(d=3^b\)-dimensional space}.
Together with Lemma~\ref{lem:block-all}, this gives the coherent information
\begin{equation}
\label{eq:sm-block-coherent-information}
I(A\rangle B)_{\rho_b}
=S(\rho_{b,B})-S(\rho_b)
=1-3^{-b}.
\end{equation}
The hashing inequality \cite{devetak2005distillation} gives
\(E_{\LOCC}^{D}(\rho_b)\ge1-3^{-b}\).  Distilling the \(r_N\) blocks
independently yields
\begin{equation}
\label{eq:sm-amplified-locc-distillation}
E_{\LOCC}^{D}(\widehat\rho_N)\ge r_N(1-3^{-b_N})=\Omega(N/\log N).
\end{equation}
For the { \(\LSCC\)} upper bound, set
{
\[
\widehat\sigma_N=\sigma_{b_N}^{\otimes r_N}
\otimes(\ket{00}\bra{00})^{\otimes(N-b_Nr_N)}.
\]
}
This reference is separable and PWF.  Applying Lemma~\ref{lem:pwf-regularized}
{ to \(\widehat\rho_N\) and \(\widehat\sigma_N\)} and then using
Lemma~\ref{lem:block-stab} gives
{
\begin{equation}
\label{eq:sm-amplified-regularized-stabilizer-bound}
\limsup_{M\to\infty}\frac1M E_{\STAB}(\widehat\rho_N^{\otimes M})
\le D_{\rm KL}(W_{\widehat\rho_N}\Vert W_{\widehat\sigma_N})
\le r_N\frac{b_N\log 3}{3^{b_N}}
\le \frac{\log 3}{N}.
\end{equation}
}
The asymptotic converse, Lemma~\ref{lem:asymp-converse}, therefore gives
\(E_{\LSCC}^{D}(\widehat\rho_N)\le(\log 3)/N\).
\end{proof}
\end{document}